\documentclass[runningheads,orivec]{llncs}

\usepackage[T1]{fontenc}

\usepackage{booktabs}
\usepackage{multirow,makecell}
\usepackage[dvipsnames]{xcolor}
\usepackage[linesnumbered,ruled,vlined]{algorithm2e}
\usepackage{comment}
\usepackage{subfig}  
\usepackage{multicol}
\usepackage{tabularx}
\usepackage{url}
\usepackage{wrapfig}
\usepackage{amsfonts}
\usepackage{enumerate}
\usepackage{enumitem}
\usepackage{MnSymbol} 
\usepackage{thm-restate}
\usepackage{fullpage}
\usepackage{xspace}
\usepackage{graphicx}
\usepackage[mode=buildnew]{standalone}
\usepackage{hyperref}
\usepackage{cleveref}

\newcommand{\ILPRobust}{MinPathError\xspace}
\newcommand{\ILPLQ}{LeastSquares\xspace}

\newcommand{\ext}[1]{\mathsf{extension}(#1)}

\title{Safe Sequences via Dominators in DAGs\\ for Path-Covering Problems}
\titlerunning{Safe Sequences via Dominators for Path Covers}

\author{Francisco Sena\inst{1}\orcidID{0000-0002-3508-4473} \and
Romeo Rizzi\inst{2}\orcidID{0000-0002-2387-0952} \and
Alexandru I. Tomescu\inst{1}\orcidID{0000-0002-5747-8350}}
\authorrunning{F. Sena, R. Rizzi, and A.I. Tomescu}

\institute{Department of Computer Science, University of Helsinki, Finland \\
\email{\{francisco.sena, alexandru.tomescu\}@helsinki.fi} 
\and
Department of Computer Science, University of Verona, Italy \\
\email{romeo.rizzi@univr.it}}

\begin{document}

\maketitle

\begin{abstract}
A \emph{path-covering} problem on a directed acyclic graph (DAG) requires finding a set of source-to-sink paths that cover all the nodes, all the arcs, or subsets thereof, and additionally they are optimal with respect to some function. In this paper we study \emph{safe sequences} of nodes or arcs, namely sequences that appear in some path of every path cover of a DAG.
We show that safe sequences admit a simple characterization via cutnodes. Moreover, we establish a connection between maximal safe sequences and leaf-to-root paths in the source- and sink-dominator trees of the DAG, which may be of independent interest in the extensive literature on dominators. With dominator trees, safe sequences admit an $O(n)$-size representation and a linear-time output-sensitive enumeration algorithm running in time $O(m + o)$, where $n$ and $m$ are the number of nodes and arcs, respectively, and $o$ is the total length of the maximal safe sequences.
We then apply maximal safe sequences to simplify Integer Linear Programs (ILPs) for two path-covering problems, \ILPLQ and \ILPRobust, which are at the core of RNA transcript assembly problems from bioinformatics. On various datasets, maximal safe sequences can be computed in under 0.1 seconds per graph, on average, and ILP solvers whose search space is reduced in this manner exhibit significant speed-ups. For example on graphs with a large width, average speed-ups are in the range $50$--$250\times$ for \ILPRobust and in the range $80$--$350\times$ for \ILPLQ. Optimizing ILPs using safe sequences can thus become a fast building block of practical RNA transcript assembly tools, and more generally, of path-covering problems.
\end{abstract}

\newpage

\section{Introduction}
\label{sec:introduction}

Data reduction, namely simplifying the input to a problem while maintaining equivalent solutions, has proved extremely successful in solving hard problems in both theory and practice~\cite{DBLP:journals/sigact/GuoN07}. While the most well-known approach for data reduction is \textit{kernelization}~\cite{DBLP:conf/birthday/LokshtanovMS12,DBLP:journals/corr/abs-1811-09429}, there are also other ways to pre-process the input in order to obtain practical algorithms. 

For example, for graph problems where a solution is a set of vertices with certain properties, Bumpus et al.~\cite{essential_vertices_1} and Jansen and Verhaegh~\cite{essential_vertices_2} considered the notion of \textit{$c$-essential vertex} as one belonging to \textit{all} $c$-approximate solutions to the problem.
They showed that for some vertex-subset finding problems (e.g.~Vertex Cover) and some values of $c$, there exist algorithms with an exponential dependency only in the number of \textit{non-$c$-essential vertices}.

In this paper we pursue a similar data reduction approach, via the notion of \textit{safe partial solution}~\cite{omnitigs_tomescu}, defined as one ``appearing'' in all solutions to a problem. For solutions that are sets of paths satisfying certain properties, a \textit{safe path} can be defined as one appearing as a subpath of some path of any solution. Intuitively, safety corresponds to $1$-essentiality.

For an NP-hard problem, it is usually intractable also to compute all its safe partial solutions (see e.g.~\cite{dias2023safety}). As such, in this paper we use the following property about safety. Let $\mathcal{S}$ be the (unknown) set of solutions to a problem, and let $\mathcal{S'} \supseteq \mathcal{S}$ be a superset of $\mathcal{S}$ whose safe partial solutions can be computed efficiently. Then, if a partial solution appears in all solutions in $\mathcal{S'}$, it also appears in all solutions in $\mathcal{S}$. 

This property was used by Grigorjew et al.~\cite{acceleratingILP} when applying safety as a pre-processing step for the Minimum Flow Decomposition (MFD) problem. The MFD problem asks to decompose the flow in a directed acyclic graph (DAG) into a minimum number of weighted source-to-sink paths~\cite{ahuja1993network,vatinlen2008simple}. As solution superset $\mathcal{S}'$, they considered flow decompositions with any number of paths, whose safe paths can be computed in polynomial time~\cite{safe_paths_fd,khan_optimizing,ahmed2024evaluating}. Instead of simplifying the input graph, they used the safe paths to simplify an integer linear program (ILP) for the MFD problem from~\cite{dias2022fast}, thus reducing the search space of the ILP solver. This resulted in speed-ups of up to 200 times on the hardest instances. 

In this paper we focus on problems whose input consists of an \emph{$s$-$t$ DAG}\footnote{By an \emph{$s$-$t$ directed graph} we mean one with a unique source $s$ and unique sink $t$ such that all nodes are reachable from $s$ and all nodes reach $t$; if the graph is acyclic, then it is an \emph{$s$-$t$ DAG}. Path-covering problems on DAGs typically require that paths either start in a source of the graph, or more generally, in a given set of starting nodes, and analogously end in a sink, or in a given set of terminal nodes. However, these start/end nodes can naturally be connected to $s$ and $t$, and thus we can assume, without loss of generality, that our DAGs are $s$-$t$ DAGs.} $G$ where the solution superset $\mathcal{S}'$ contains all the sets of paths in $G$ from $s$ to $t$ that together cover all the nodes or arcs of $G$ (or given subsets thereof). We will refer to these objects simply as \emph{path covers}. Note that a flow decomposition is also a path cover of the arcs carrying flow, but path covers generalize many other problems, especially those modeling practical problems from bioinformatics, see \Cref{sec:scaling-ILPs-contrib}.
Next, we present the contributions of the paper and further contextualization.

\subsection{Optimal enumeration of maximal safe sequences via dominators}

\paragraph*{Safe sequences as a generalization of safe paths.} The notion of a safe \textit{path} has been extensively used to capture \textit{contiguous} safe partial solutions (see e.g.~\cite{omnitigs_tomescu,schmidt2023omnitig,dias2023safety,obscura2018safe,rahman2022assembler,khan_optimizing,genome_assembly_from_practice_to_theory}). However, for our purposes contiguity is not necessary, but only the fact that the nodes/arcs in a safe partial solution appear in the same order in some path of every solution. For example, it may be that the only way to reach a node $u$ is via a node $v$, with some complex subgraph between them. The sequence $u,v$ thus appears in some path of every path cover of the nodes.

We thus define a \emph{safe sequence} with respect to the path covers of a DAG $G$, either of the nodes or of the arcs, analogously to that of a safe path (see \Cref{fig:seq-example} for an example).
This notion is inspired by the work of Ma, Zheng and Kingsford~\cite{and_or_quant}, who proposed a problem (AND-Quant) and a max-flow-based algorithm deciding when a sequence of arcs appears in some path in every flow decomposition (with any number of paths) of a flow in a DAG. Here, we are particularly interested in characterizing but also in computing \emph{maximal} safe sequences, i.e., safe sequences that are not a proper subsequence of another safe sequence.

\begin{figure}[t]
    \centering
    \includegraphics[width=1\linewidth]{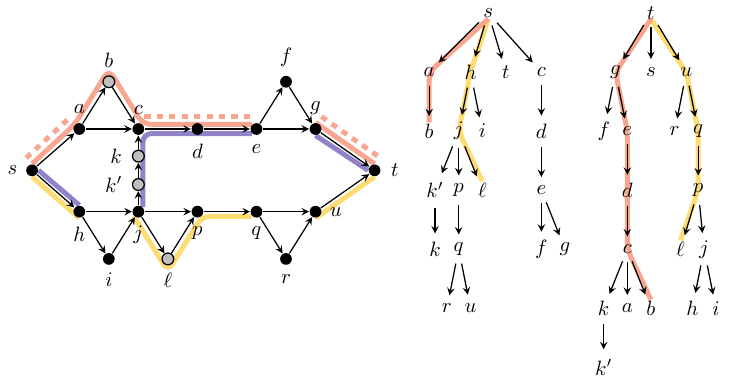}
    \caption{An $s$-$t$ DAG and its $s$- and $t$-dominator trees. Three maximal safe sequences with respect to path covers of the nodes are shown in solid orange, violet, and yellow. Nodes $b$ and $\ell$ are leaves in both dominator trees, and thus their extension into the orange and yellow sequences, resp., are maximally safe. (The maximal safe sequences obtained by extending nodes $f$, $i$ and $r$ are not shown.) Note that $k'k$ is a unitary path and that $\ext{k'}=\ext{k}$. The dashed orange safe sequence is not maximal and can be obtained by extending $a$.
    }
    \label{fig:seq-example}
\end{figure}

\paragraph*{Characterization of safe sequences.} Assume that the input DAG has $n$ nodes and $m$ arcs. We show that safe sequences admit a simple characterization in terms of cutnodes. A \textit{$u$-$v$ cutnode} is a node such that all paths from $u$ to $v$ (\textit{$u$-$v$ paths}) pass through it.
We present our results for path covers of the nodes and refer to~\Cref{sec:arc-path-covers} for the results on path covers of the arcs. Omitted proofs can be found in~\Cref{sec:additional-proofs}.

\begin{theorem}[Characterization of safe sequences]
\label{thm:safe_sequences_char}
    Let $G=(N,A)$ be an $s$-$t$ DAG and $X$ be a sequence of nodes. The following statements are equivalent:
    \begin{enumerate}
        \item $X$ is safe for path covers.
        \item There is a node $u \in N$ such that every $s$-$t$ path covering $u$ contains $X$.
        \item There is a node $v \in N$ such that $X$ is contained in the sequence obtained as the concatenation of the sequences of $s$-$v$ and $v$-$t$ cutnodes.
    \end{enumerate}
\end{theorem}

While outputting the sequences of $s$ and $t$ cutnodes for every node $v$ in the graph can be done in time $O(m)$ per node, different nodes can thus be extended to the same maximal safe sequence, and some extensions may not even be maximal (see \Cref{fig:seq-example}). The challenge is to characterize those nodes identifying all and only maximal safe sequences, without duplicates.

\paragraph*{Linear-time output-sensitive enumeration via dominators.} We show that we can identify such nodes in an elegant manner using \emph{dominators}, which is a well-studied concept related to cutnodes (see e.g.~the survey \cite{parotsidis2013dominators} for an in-depth contextualization of dominators). In particular, the \emph{$s$-dominator tree} of an $s$-$t$ DAG $(N,A)$ has node set $N$, and $u$ is the parent of $v$ iff all paths from $s$ to $v$ pass through $u$, and moreover there is no node reachable from $u$ (except $u$ itself) whose removal makes $v$ unreachable from $u$.

The key idea behind our results is to work \emph{also} with the $t$-dominator tree of $G$, which is defined in an analogous manner, but considering reachabilities \emph{to} $t$. By analyzing the interplay between the two trees, we show that the nodes solving this problem (\Cref{thm:cores-nodes}) are those that are \emph{leaves in both} the $s$- and $t$-dominator trees, only after overcoming some technicalities arising from unitary paths (see~\Cref{fig:seq-example},~\Cref{lem:compressed-graph} and~\Cref{lemma:no_redundancy_nodes}). Since dominator tree leaves admit a characterization in terms of in- or out-neighborhoods on DAGs (\Cref{lemma:node_dominator}), this leads to an $O(mn)$-time algorithm outputting all and only maximal safe sequences without duplicates, independently of dominator trees: 

\begin{restatable}{theorem}{mnallonlymaximal}
\label{thm:mn-all-only-maximal}
    Given an $s$-$t$ DAG $G$ with $n$ nodes and $m$ arcs, we can compute all and only the maximal safe sequences of $G$ in $O(nm)$-time, without constructing dominator trees.
\end{restatable}

To improve this time bound we resort to the fact that dominator trees admit an $O(m+n)$ time construction \cite{alstrup1999dominators,buchsbaum2008linear,fraczak2013finding}, although so far some of these algorithms do not seem to be practical. See~\cite{georgiadis2006finding,georgiadis2014loop,cooper2001simple} for experimental studies on different dominator tree algorithms. Regardless, any dominator algorithm serves our purposes and its effect on the runtime of our algorithm is clear in the proof of the theorem below.

\begin{restatable}[Enumeration of maximal safe sequences]{theorem}{maxsafeseqsenumrepr}
\label{thm:representation-optimal-enumeration}
    Let $G$ be an $s$-$t$ DAG with $n$ nodes and $m$ arcs. The following hold:
    \begin{enumerate}
        \item There is an $O(m+o)$ time algorithm outputting the set of all maximal safe sequences with no duplicates, where $o$ denotes the total length of all the maximal safe sequences.
        \item There is an $O(n)$-size representation of all the maximal safe sequences of $G$ that can be built in $O(m+n)$ time.
    \end{enumerate}
\end{restatable}

An analogous result to \Cref{thm:representation-optimal-enumeration} can be obtained for path covers where only a given subset $C \subseteq N$ of nodes has to be covered. To achieve that, we have to analyze the dominator trees induced by the dominance relation restricted to $C$. Moreover, analogous results can be obtained also for path covers of the arcs. The details can be found in~\Cref{sec:subset-path-covers}.

\subsection{Scaling path-covering ILPs for bioinformatics problems}
\label{sec:scaling-ILPs-contrib}

The motivation of our work comes from bioinformatics, where many NP-hard graph problems related to sequencing data require finding a set of paths in a DAG satisfying certain properties. For example, in the \textit{RNA transcript assembly} problem, we need to recover the RNA transcripts produced in various abundances by a gene, given only smaller fragments (i.e.~RNA-seq reads) sequenced from the RNA transcripts. This problem is usually solved by first building an arc-weighted DAG from the fragments, i.e., a \emph{splice graph}. The RNA transcripts correspond to a set of source-to-sink weighted paths in the DAG that ``best explain'' the arc-weights, under various definitions of optimality, see e.g.~\cite{translig,cliiq,multitrans,ssp,nsmap,jumper,class2,isoinfer,dias2022fast,dias2023safety}. 

Motivated by this, we apply safe sequences in two path-covering problems defined on DAGs that are at the core of practical tools for RNA transcript assembly. We explain these in \Cref{sec:encode-safety}, and briefly describe them below. 

The first problem, \ILPLQ, is used by several RNA assembly tools e.g.~IsoLasso~\cite{isolasso}, CIDANE~\cite{cidane}, Ryuto~\cite{ryuto}, SLIDE~\cite{slide}, iReckon~\cite{ireckon}, TRAPH~\cite{traph}.

Here, we need to find a set of $k$ paths minimizing the sum of the squared errors between the weight of each arc and the weight of the solutions paths going through the arc. Formally, given an $s$-$t$ DAG $G = (N,A)$, a positive weight $w(uv)$ for every arc $uv \in A$, and $k \geq 1$, we need to find $s$-$t$ paths $P_1,\dots,P_k$, with associated weights $w_1,\dots,w_k$, minimizing 
\[\sum_{uv \in A} \left(w(uv) - \sum_{i ~:~ uv \in P_i}w_i\right)^2.\]
While solution paths $P_1,\dots,P_k$ many not cover all arcs, we can heuristically infer that they cover the arcs with large weight.

The second problem, \ILPRobust, was recently introduced in~\cite{dias2024robust} and shown to be more accurate than \ILPLQ. The goal here is to account for errors at the level of solution paths, and minimize the sum of the errors of the paths. Formally, we need to find weighted paths as above, but we also need to associate a slack $\rho_i$ to each path minimizing $\sum_{i=1}^{k}\rho_i$, and such that the following condition holds:
\[\left|w(uv) - \sum_{i ~:~ uv \in P_i}w_i\right| \leq \sum_{i ~:~ uv \in P_i}\rho_i, ~~~\forall uv \in A.\]
By definition, solution paths $P_1,\dots,P_k$ need to cover every arc $uv$, since otherwise the above condition becomes $w(uv) \leq 0$, which contradicts the assumption that $w(uv)$ is positive.

We use safe sequences to simplify ILPs for these two problems, by fixing suitable ILP variables encoding paths, as in~\cite{acceleratingILP}. On various datasets, maximal safe paths can be computed in under 0.1 seconds per graph, on average, and such safety-optimized ILPs show significant speed-ups over the original ones. More specifically, on the harder instances of graphs with a large width, average speed-ups are in the range $50-250\times$ for \ILPRobust and in the range $80-350\times$ for \ILPLQ. As such, our optimizations can become a scalable building block of practical RNA transcript assembly tools, but also of general path-covering problems.

\section{Notation and preliminaries}
\label{sec:notation-preliminaries}

\paragraph*{Graphs.} In this paper we consider directed graphs, especially directed acyclic graphs (DAGs). For a DAG $G=(N,A)$ with node set $N$ and arc set $A$, we let $n=|N|$ and $m=|A|$. We denote an arc from $u$ to $v$ as $uv$ and a path $p$ through nodes $v_1,\dots,v_k$ as $p=v_1\dots v_k$. Two paths are considered distinct if they differ in at least one node.
The set of \emph{in-neighbors} (\emph{out-neighbors}) of a node $v$ is denoted $N^-_v$ ($N^+_v$), and its \emph{indegree} (\emph{outdegree}) as $d^-_v$ ($d^+_v$). For nodes $u,v$ we say that there is a \emph{$u$-$v$ path} if there is a path whose first node is $u$ and last node is $v$ (we also say that $u$ \emph{reaches} $v$). An \emph{antichain} is a set of pairwise unreachable nodes. The size of the largest antichain in a DAG is called \emph{width}. Analogously, for arcs $uv,xy$ we say that $uv$ reaches $xy$ if $v$ reaches $x$. An \emph{arc-antichain} is a set of pairwise unreachable arcs, and the \emph{arc-width} of a DAG is the size of the largest arc-antichain of the graph.

For $X \subseteq N$ we denote a \emph{sequence} $X$ through nodes $u_1,\dots,u_\ell$ as $X=u_1,\dots,u_\ell$ if there is a $u_i$-$u_{i+1}$ path for all $1 \le i \le \ell-1$; in particular, a path is a sequence. We say that $X$ \emph{covers} or \emph{contains} a node $v$ if $v=u_i$ for some $i \in \{1, \dots, \ell\}$, and we write $u \in X$. If $Y$ is a sequence and each node of $X$ is contained in $Y$ we say that $X$ is a \emph{subsequence} of $Y$, or that $X$ is \emph{contained} in $Y$. Since $G$ is a DAG, the nodes of $X$ appear in the same order as in $Y$.

Let $X'=v_1,v_2,\dots,v_{\ell'}$ be a sequence such that there is a path from the last node $u_\ell$ of $X$ to $v_1$. The \emph{concatenation} of $X$ and $X'$ is simply the sequence $XX'$ obtained as $X$ followed by $X'$. If $u_\ell=v_1$, then the repeated occurrence of $v_1$ is removed, i.e. $XX' := u_1,\dots,u_\ell,v_2,\dots,v_{\ell'}$.
We adopt the convention that lowercase letters $u,v,w,x,y,z$ denote nodes and uppercase letters $X,Y$ denote sequences of nodes.

\paragraph*{Safety.} We define a \emph{path cover} of $G$ to be a set $P$ of $s$-$t$ paths such that $\forall u \in N ~ \exists p \in P : u \in p$; we also say \emph{cover} instead of ``$s$-$t$ path cover of $G$'' when the graph $G$ is clear from the context. Two path covers are distinct if they differ in at least one path. We assume that every node of $G$ lies in some $s$-$t$ path.
We say that a sequence $X$ is \emph{safe} for $G$ if in every path cover $P$ of $G$ it holds that $X$ is a subsequence of some path in $P$.
If $X$ is not safe we say that it is \emph{unsafe}. A sequence $X$ is \emph{maximally safe} if it is not a subsequence of any other safe sequence. We assume that there is always an $s$-$t$ path containing a given sequence of nodes, otherwise that sequence is vacuously unsafe.
Note that imposing the collection of $s$-$t$ paths of a cover to be minimal with respect to inclusion results in an equivalent definition of safety: any sequence that is safe for path covers is also safe for minimal path covers since every minimal path cover is a path cover; but also, if $X$ is safe for minimal path covers then it is also safe for general path covers, as otherwise we could have a path cover $P$ avoiding $X$ from which arbitrary $s$-$t$ paths can be removed until $P$ becomes minimal.

\paragraph*{Graph compression.} We define a \emph{unitary path} as a path where each node has indegree equal to one but the first one, and each node has outdegree equal to one but the last one. For example, in \Cref{fig:seq-example}, the paths $de$ and $pq$ are unitary, the path $cde$ is maximally unitary. With respect to $G$, we define the \emph{compressed graph of $G$} as the graph where every maximal unitary path of $G$ is contracted into a single node. It is a folklore result that compressing a graph can be done in linear time (see e.g.~\cite{chikhi2016compacting,omnitigs_tomescu}). We state this fact in the next lemma together with an additional property stating that compressed graphs cannot have arcs $uv$ with $v$ the unique out-neighbor of $u$ and $u$ the unique in-neighbor of $v$, since $uv$ would then be a unitary path.
\begin{lemma}
\label{lem:compressed-graph}
Let $G$ be a directed graph with $n$ nodes and $m$ arcs. Then the following hold:
\begin{enumerate}
\item We can compute the compressed graph of $G$ in $O(n+m)$ time.\label{lem:compressed-graph-construction}
\item In a compressed graph there is no arc $uv$ such that $N^+_u = \{v\}$ and $N^-_v = \{u\}$.\label{lem:compressed-graph-arcs}
\end{enumerate}
\end{lemma}

The compression of unitary paths preserves safety: every $s$-$t$ path in $G$ traversing the unitary path traverses the corresponding node in the compressed graph of $G$, and every $s$-$t$ path traversing a node in the compressed graph of $G$ necessarily traverses the unitary path it corresponds to in $G$. This idea appeared in the context of walk-finding problems in \cite{genome_assembly_from_practice_to_theory}.

\paragraph*{Cutnodes and dominators.} 

In the dominator literature, the directed graphs on which dominators are computed have a single start node $s$ from which every node is reachable. In some cases, the graphs have instead or additionally a unique sink $t$ which is reached by every node, in which case the dominance relation is referred to as post-dominance (see e.g., \cite{allen1970control,gupta1992generalized}). A prominent application of dominators is in the context of program analysis, where $s$ corresponds to the entry point of the program and $t$ to its exit point. Our graphs always have unique abstract nodes $s$ and $t$, hence, diverging from standard notation, we parameterize the domination relation by $s$ and $t$. We borrow standard concepts and notation from the literature of dominators, which we describe next.

Let $G$ be an $s$-$t$ directed graph. A node $u$ is said to \emph{$s$-dominate} a node $v$ if every $s$-$v$ path contains $u$, hence $u$ is an $s$-$v$ cutnode if and only if $u$ $s$-dominates $v$. Every node $s$- and $t$-dominates itself. We also say that $v$ \emph{$t$-dominates} $u$ if every $u$-$t$ path contains $v$, which is equivalent to $v$ being a $u$-$t$ cutnode.
We say that $u$ \emph{strictly $s$-dominates} $v$ if $u$ $s$-dominates $v$ and $u \neq v$. For $v \neq s$, the \emph{immediate $s$-dominator} of $v$ is the strict $s$-dominator of $v$ that is dominated by all the nodes that strictly dominate $v$, in which case we write $idom_s(v)=u$.
The domination relation on $G$ with respect to $s$ can be represented as a rooted tree at $s$ with the same nodes as $G$, called the \emph{$s$-dominator tree} of $G$. In this tree, the parent of every node is its immediate $s$-dominator, every node is $s$-dominated by all its ancestors and $s$-dominates all its descendants and itself. A node is said to be an \emph{$s$-dominator} if it strictly $s$-dominates some node.
The same terminology will be used for the $t$-domination relation.
We also define a function $dom: N \times \mathbb{N}^+ \to N$ on rooted trees as follows. Let $dom_s(v,k) := idom_s(v)$ when $k=1$ and $dom_s(v,k) := dom_s(dom_s(v,1),k-1)$ when $k>1$. If $k$ is larger than the number of \emph{strict} dominators of a given node, then it defaults to $s$ (resp.~$t$) in the $s$-dominator tree (resp. $t$-dominator tree). Essentially, $dom$ gives the $k$-th ancestor of a node in a rooted tree, if it is well defined, otherwise it returns the root of the tree. For example, in \Cref{fig:seq-example}, the immediate $s$-dominator of $b$ is $a$ ($dom_s(b,1)=a$) and $dom_t(q,3)=t$. Node $c$ is a strict $t$-dominator of the nodes $\{a,b,k,k'\}$. See~\cite{parotsidis2013dominators} for an in-depth presentation of dominators.

In DAGs, dominators admit a characterization, which we state in \Cref{lemma:node_dominator} below. This was first presented by Ochranov{\'a}~\cite{ochranova1983finding}, but with an incomplete proof. This characterization was also discussed by Alstrupt et al.~\cite[Sec.~5.1]{alstrup1999dominators}. To the best of our knowledge, we did not find a complete proof of this characterization in the literature, and hence we give one in~\Cref{sec:additional-proofs}.

\begin{restatable}[Characterization of node-dominators in DAGs]{lemma}{dominatorscharacterization}
\label{lemma:node_dominator}
    Let $G=(N,A)$ be an $s$-$t$ DAG. A node $u \in N$ is an $s$-dominator (resp. $t$-dominator) if and only if there is a node $v$ such that $N^-_v=\{u\}$ (resp. $N^+_v=\{u\}$).
\end{restatable}

\section{Safe sequences via dominators}
\label{sec:safe-sequences-dominators}

\subsection{Characterization of safe sequences}
\label{sec:characterization}

We start by proving \Cref{thm:safe_sequences_char}.
\begin{proof}[Proof of \Cref{thm:safe_sequences_char}]
    $(1 \Rightarrow 2)$ If for every node $u \in N$ there is an $s$-$t$ path covering $u$ that avoids $X$, then we can build a path cover by considering for every node $u$ of $G$ an $s$-$t$ path avoiding $X$ that covers $u$. This contradicts the safety of $X$.
    
    $(2 \Rightarrow 3)$ Since every $s$-$t$ path covering $u$ contains the sequence of $s$-$u$ cutnodes followed by the sequence of $u$-$t$ cutnodes and $X$ is contained in every such path, the implication follows by taking $u=v$.

    $(3 \Rightarrow 1)$ Since $X$ is a subsequence of a sequence containing every $s$-$v$ and $v$-$t$ cutnode and $v$ must be covered in every path cover by some $s$-$t$ path, any of which contains the $s$- and $t$-cutnodes of $v$, it follows that $X$ is safe.\hfill $\square$
\end{proof}

We will heavily use the type of sequence obtained as in point 3 of \Cref{thm:safe_sequences_char}. As such, we define the \emph{extension} of a node $v$, $\ext{v}$, as the sequence obtained from the concatenation of the sequence of $s$-$v$ cutnodes with the sequence of $v$-$t$ cutnodes. By \Cref{thm:safe_sequences_char} every maximal safe sequence is the extension of some node, and so in a DAG with $n$ nodes and $m$ arcs there are at most $n$ maximal safe sequences. Moreover, for any node $v \in N$, $\ext{v}$ can be computed in time $O(m+n)$ by computing the $s$-$v$ cutnodes and $v$-$t$ cutnodes with e.g.~the simple algorithm of Cairo et al.~\cite{CAIRO2021103}. Thus, \Cref{thm:safe_sequences_char} implies that we can compute in time $O(nm)$ a set of safe sequences containing every maximal safe sequence of $G$, by reporting $\ext{v}$ for every node $v$. However, the reported sequences may not be maximal and the same sequence may be reported multiple times.

Because of the correspondence between cutnodes and dominators explained in \Cref{sec:notation-preliminaries}, the extension of a node can also be expressed as suitable paths in dominator trees. We will use this fact in the rest of the paper.

\begin{remark}
\label{remark:extension-and-dominators}
    In an $s$-$t$ DAG $G$, $\ext{v}$ is the same as the path from $s$ to $v$ in the $s$-dominator tree of $G$ concatenated with the path from $v$ to $t$ in the $t$-dominator tree of $G$.
\end{remark}

\subsection{Properties of s- and t-dominator trees}

We start with a lemma showing an interplay between $s$-domination and $t$-domination on general $s$-$t$ directed graphs.

\begin{lemma}
\label{lem:st-dom-immediate-relation}
    Let $G=(N,A)$ be an $s$-$t$ directed graph, let $u,v \in N$ be nodes, and let $k \in \mathbb{N}^+$. If $u$ is the $k$-th ancestor of $v$ in the $s$-dominator tree, then $dom_t(u,k)$ $t$-dominates $v$, and $v$ is not a $t$-dominator of $u$ unless $v=dom_t(u,k)$.
\end{lemma}
\begin{proof}
    Note that $u \neq t$ since $u$ is an $s$-dominator. Let $w=dom_t(u,k)$ and let $y_1,\dots,y_{k-1}$ be the nodes between $u$ and $v$ in the $s$-dominator tree. We can assume $w \neq v$.

    Observe that $w \neq y_i$ for all $i \in \{1,\dots,k-1\}$. Suppose otherwise. Then $w \neq t$ because $t$ is not an $s$-dominator. Hence, by definition of $dom_t$, there are $k-1$ nodes $x_1,\dots,x_{k-1}$ between $w$ and $u$ in the $t$-dominator tree. Some $x_j$ is not between $u$ and $w$ in the $s$-dominator tree, and hence there is a $u$-$w$ path avoiding $x_j$. Further, since $x_j$ is strictly $t$-dominated by $w$ there is a $w$-$t$ path avoiding $x_j$. Concatenating these paths at $w$ results in a $u$-$t$ path avoiding $x_j$, contradicting the fact that $x_j$ $t$-dominates $u$.

    Suppose now for a contradiction that $v$ is not strictly $t$-dominated by $w$. From the point above, there is a $u$-$v$ path avoiding $w$. This path concatenated with a $v$-$t$ path avoiding $w$ gives a $u$-$t$ path also avoiding $w$, contradicting the fact that $w$ $t$-dominates $u$.
    
    It remains to show that $v$ cannot $t$-dominate $u$. Suppose otherwise. Then we have $v=x_i$ for some $i \in \{1,\dots,k-1\}$, since $v$ is strictly $t$-dominated by $w$ and $t$-dominates $u$. Thus, some $y_j$ is not between $v$ and $u$ in the $t$-dominator tree, implying that there is a $u$-$v$ path avoiding $y_j$. Since $y_j$ is $s$-dominated by $u$, there is an $s$-$u$ path avoiding $y_j$, which concatenated with the previous path gives an $s$-$v$ path also avoiding $y_j$, contradicting the fact that $y_j$ $s$-dominates $v$. \hfill $\square$  
\end{proof}

\Cref{lem:st-dom-immediate-relation} can be interpreted for the special case of $k=1$ as follows: if $u$ immediately $s$-dominates $v$ then the immediate $t$-dominator of $u$ $t$-dominates $v$ (see~\Cref{fig:proof-sketch}). Another consequence of this result is that if $u$ is an $s$-dominator, then at most one node immediately $s$-dominated by $u$ is in the path from $t$ to $u$ in the dominator tree of $t$, which must be the immediate $t$-dominator of $u$.

Observe now in \Cref{fig:seq-example} that $c$ is the immediate $s$-dominator only of $d$ and that $d$ is the immediate $t$-dominator only of $c$, and likewise for $d$ and $e$, and so $\ext{c}=\ext{d}=\ext{e}$. This redundancy and asymmetry introduced by unitary paths is undesirable in the characterization of maximal safe sequences. This discussion motivates the following result.

\begin{lemma}
\label{lemma:no_redundancy_nodes}
    Let $G=(N,A)$ be a compressed $s$-$t$ DAG and let $u,v\in N$ be distinct nodes. If $u$ immediately $s$-dominates $v$ and $v$ immediately $t$-dominates $u$ then $u$ immediately $s$-dominates some node different from $v$ and $v$ immediately $t$-dominates some node different from $u$.
\end{lemma}
\begin{proof}
    Assume without loss of generality that $u$ immediately $s$-dominates only node $v$.
    
    We claim that $N^-_v=\{u\}$. Suppose otherwise. If $N^-_v$ is empty then $v=s$, and so $v$ is not a $t$-dominator, a contradiction. So there is an arc $wv$ with $w \neq u$. Since $u$ $s$-dominates $v$, it must also $s$-dominate every in-neighbour of $v$, namely $w$. But $u$ immediately $s$-dominates only $v$, thus, $v$ must $s$-dominate $w$, implying that there is a $v$-$w$ path. Since $wv \in A$, $G$ contains a cycle, a contradiction.

    Since $v$ is a $t$-dominator, by \Cref{lemma:node_dominator} there is a node whose unique out-neighbour is $v$. If this node is $u$, we contradict the fact that the graph is compressed by \Cref{lem:compressed-graph}. Therefore, $v$ must have at least two distinct in-neighbors, a contradiction.\hfill $\square$
\end{proof}

\begin{figure}[t]
    \centering
    \includegraphics{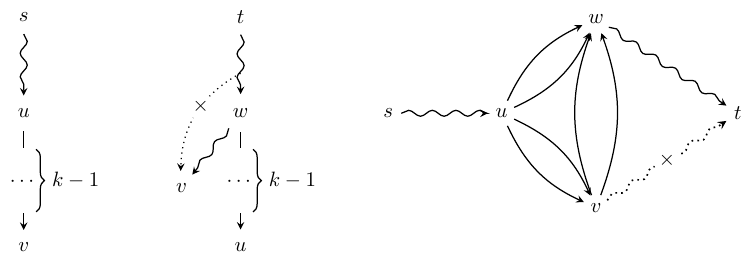}
    \caption{Illustration of the proof of \Cref{lem:st-dom-immediate-relation}. The two dominator trees discussed in the proof are shown on the left. On the right, an abstract graph for the special case of $k=1$ is shown; node $v$ cannot reach $t$ without $w$, without contradicting the fact that $dom_t(u,1)=w$.}
    \label{fig:proof-sketch}
\end{figure}

\subsection{Characterization and enumeration of maximal safe sequences}

Using \Cref{lem:st-dom-immediate-relation,lemma:no_redundancy_nodes}, we can give a precise characterization of maximal safe sequences.

\begin{theorem}[Characterization of maximal safe sequences]
\label{thm:cores-nodes}
    Let $G=(N,A)$ be a compressed $s$-$t$ DAG and let $u \in N$ be a node. Then $\ext{u}$ is a maximal safe sequence if and only if $u$ is a leaf in both the $s$- and $t$-dominator trees of G.
\end{theorem}
\begin{proof}
    $(\Leftarrow)$ If $u$ is a leaf in both dominator trees then $\ext{u}$ is maximal: no other node can be added to this sequence by the equivalence of safe sequences and paths in the dominator trees given by \Cref{remark:extension-and-dominators,thm:safe_sequences_char}.

    $(\Rightarrow)$ Suppose that $\ext{u}$ is a maximal safe sequence but $u$ is not a leaf in both dominator trees. Without loss of generality suppose that $u$ is not a leaf in the $s$-dominator tree. Note that $u \neq t$ since $u$ is an $s$-dominator. Let $v$ be a node immediately $s$-dominated by $u$ and let $w$ denote the immediate $t$-dominator of $u$.
    
    First we argue that there is a node $z$ that is immediately $s$-dominated by $u$ and that is strictly $t$-dominated by $w$. We consider two cases. If $v \neq w$ then we can take $z = v$, because by \Cref{lem:st-dom-immediate-relation} for $k=1$ we can conclude that $w$ $t$-dominates $v$, and since $v \neq w$, it strictly $t$-dominates $v$. If $v = w$ then we can take $z=v'$, since by \Cref{lemma:no_redundancy_nodes} there is a node $v' \neq v$ that is immediately $s$-dominated by $u$; applying \Cref{lem:st-dom-immediate-relation} to $u$ and $v'$ with $k=1$, we conclude that $w$ $t$-dominates $v$, and since $w=v\neq v$, $w$ strictly $t$-dominates $v'$.
    
    Now we show that $\ext{z}$ properly contains $\ext{u}$.
    Consider $\ext{u}$, which by definition consists of the sequence of $s$-$u$ cutnodes (call it $Y_1$) concatenated with the sequence of $u$-$t$ cutnodes (call it $Y_2$). In $\ext{u}$, by definition, the double occurrence of $u$ is removed, and thus let $Y_2'$ be the sequence of $w$-$t$ cutnodes. That is, $\ext{u}$ can be written as $Y_1 Y_2'$.
    Consider also $\ext{z}$, which can be written analogously as $X_1 X_2'$.
    Since $u$ is the parent of $z$ in the $s$-dominator tree, we have that $Y_1$ is \emph{strictly} contained in $X_1$, and since $z$ is a proper descendant of $w$ in the $t$-dominator tree, we have that $Y_2'$ is contained in $X_2'$ (possibly $Y_2'=X_2'$). Therefore, $\ext{u}$ is a proper subsequence of $\ext{z}$, contradicting our assumption that $\ext{u}$ is a maximal safe sequence.\hfill $\square$
\end{proof}

The idea in the above theorem of targeting special nodes or arcs in order to characterize every maximal safe object is a recurring theme in safety problems. For instance, \cite{khan_et_al,ahmed2024evaluating} proposed the notion of \emph{representative arc}, \cite{schmidt_et_al} proposed that of \emph{core}, and \cite{macrotigs} found a special class of walks, coined \emph{macrotigs}, from where every maximal safe walk can be built. Each of these constructs was used in a different safety problem.

While the above characterization uses dominator trees, we can derive an equivalent characterization of dominator-tree leaves just in terms of in- and out-neighborhoods via \Cref{lemma:node_dominator}, since a node is a leaf if and only if it is not a dominator. Computing the extension of every node satisfying such condition leads to a simple $O(mn)$ time algorithm outputting all and only maximal safe sequences, without needing to compute the dominator trees, from which \Cref{thm:mn-all-only-maximal} follows.

Since dominator trees can be built in $O(m+n)$ time (see, e.g., the algorithm by Alstrup et al.~\cite{alstrup1999dominators}, Buchsbaum et al.~\cite{buchsbaum2008linear}), or Fraczak et al.~\cite{fraczak2013finding}), we can prove \Cref{thm:representation-optimal-enumeration}.

\begin{proof}[Proof of \Cref{thm:representation-optimal-enumeration}]
    Let $G'$ be the compressed graph of $G$. Build the dominator trees of $G'$ and compute all the leaves common to both trees. These three steps take $O(m+n)$ time.
    
    \textbf{1:} For every node that is a leaf in both dominator trees, output $\ext{v}$. By \Cref{thm:cores-nodes} we find every maximal safe sequence during this process. Further, no duplicate safe sequence is computed since different leaves common to both trees have different extensions. Therefore the algorithm runs in time equal to the length of all the maximal safe sequences plus the time required to build two dominator trees over $G'$, so $O(m+o)$ altogether.

    \textbf{2:} With the equivalence given by \Cref{thm:cores-nodes}, we have that every maximal safe sequence is the result of $\ext{v}$ for every node $v$ of $G'$ that is a leaf in both dominator trees. The sequence $\ext{v}$ is encoded in the dominator trees by \Cref{remark:extension-and-dominators}.\hfill $\square$
\end{proof}

\section{Experiments}
\label{sec:experiments}

The experiments were performed on an AMD 32-core machine with 512GB RAM. The source code of our project is publicly available on Github\footnote{\url{https://github.com/algbio/safe-sequences}}. All our algorithms are implemented in Python3, as are the scripts to produce the experimental results. To solve ILPs, we use Gurobi's Python API (version 11) under an academic license. For every graph, we run Gurobi with 4 threads and set up a timeout of 300 seconds in Gurobi's execution time.

\paragraph*{Implementation.}

As described in \Cref{sec:scaling-ILPs-contrib}, we study experimentally the potential of applying safe sequences in two arc-path covering problems, \ILPLQ and \ILPRobust. 

We did not use an external algorithm to compute arc-dominator trees because these are mostly tailored for node-dominators. Instead, we compute the immediate $s$-dominator and $t$-dominator of all arcs with an $O(m^2)$ time procedure via the algorithm of~\cite{CAIRO2021103}, appropriately handling the reversed graph of $G$ for the immediate $s$-domination relation. The classical algorithm of Aho and Ulman~\cite{aho1973theory} to build dominator trees has the same running time as our procedure. We remark that even with a suboptimal algorithm to compute the dominator trees, our entire safety-preprocessing step takes less than 0.1 seconds on average.

With the dominator trees, we can compute all the maximal safe sequences with respect to arbitrary subsets $C \subseteq A$ of arcs to be covered, all in time proportional to the total length of all the maximal safe sequences, completely analogously to the discussion in~\Cref{sec:safe-sequences-dominators}.

The ILPs were optimized as in~\cite{acceleratingILP} by fixing $x_{uvi}$ variables to 1 based on the maximal safe sequences of the input DAG, which reduces the search space of the linear program and hence speeds-up the solver. Let $X_1,\dots,X_t$ be sequences in the input DAG such that: (a) each $X_i$, $i \in \{1,\dots,t\}$, must appear in some solution path of the ILP, and (b) there exists no path in the DAG containing two distinct $X_i$ and $X_j$. In the case that every $X_i$ is a path, Grigorjew et al.~\cite{acceleratingILP} noticed that if the two mentioned properties hold, then $X_1,\dots,X_t$ must appear in $t$ distinct paths of among the $k$ solution paths of the ILP (see the three solid-colored sequences shown in~\Cref{fig:seq-example} which must appear in different paths of any path cover). As such, without loss of generality, one can set $x_{uvi} = 1$ for each arc $uv$ of $X_i$, for all $i \in \{1,\dots,t\}$. When using sequences to set binary variables, we can proceed in a completely analogous manner.
To have a large number of binary variables set to 1, we follow the approach of~\cite{acceleratingILP}. Namely, for every arc $uv$, we define its weight as the length of the longest safe sequence containing $uv$. Then we find a maximum-weight antichain of arcs in the DAG using the min-flow reduction of~\cite{rival2012graphs}. Then, for each $a_i$ in the maximum-weight antichain, we consider a longest safe sequence $X_i$ containing $a_i$ to fix variables to 1. Note that a path covering two such sequences $X_i$ and $X_j$ would witness that $a_i$ and $a_j$ cannot be in an antichain. Finally, we set $k$ to the arc-width of the graph and run Gurobi to solve the ILP.

Note that any solution to \ILPRobust is also a path cover, since each arc weight is positive, and thus the simplified ILP has equivalent optimal solutions. This may not be the case for \ILPLQ, since some arcs may not be covered by any solution path. To handle such scenarios we propose using maximal safe sequences for path covers where only a subset $C$ of nodes or arcs can be trusted to appear in any optimal solution.
As an illustration of this approach, we perform experiments by choosing the subset $C$ as those arcs whose weight is at least as large as the 25-th percentile of the weights of all arcs, since we can heuristically infer that the solution paths of \ILPLQ cover the arcs of large enough weight.

\paragraph*{Datasets.} We experiment with seven datasets. The first four contain splice graphs with erroneous weights created directly from gene annotation by~\cite{dias2024robust}.\footnote{Available at \url{https://doi.org/10.5281/zenodo.10775004}} These were created starting from the well-known dataset of Shao and Kingsford~\cite{shao2017theory}, which was also used in several other studies benchmarking the speed of minimum flow decomposition solvers~\cite{dias2022fast,kloster2018practical,williams2023subpaths}. The original splice graphs were created in~\cite{shao2017theory} from gene annotation from \textbf{Human}, \textbf{Mouse} and \textbf{Zebrafish}. This dataset also contains a set of graphs (\textbf{SRR020730}) from human gene annotation, but with flow values coming from a real sequencing sample from the Sequence Read Archive (SRR020730), quantified using the tool Salmon~\cite{salmon}. To mimic splice graphs constructed from real read abundances,~\cite{dias2024robust} added noise to the flow value of each arc, according to a Poisson distribution. The last three datasets contain graphs constructed by a state-of-the-art tool for RNA transcript assembly, IsoQuant~\cite{isoquant}. Mouse annotated transcript expression was simulated according to a Gamma distribution that resembles real expression profile~\cite{isoquant}. From these expressed transcripts, PacBio reads were simulated using IsoSeqSim~\cite{isoseqsim} and fed to IsoQuant for graph creation (we call this \textbf{Mouse PacBio}), and ONT reads were simulated using Trans-Nanosim~\cite{trans-nanosim}, which was modified to perform a more realistic read truncation (see~\cite{isoquant}), and fed to IsoQuant for graph creation (we call this \textbf{Mouse ONT}).
For the latter one, there is also a simpler version where no read truncation was introduced (\textbf{Mouse ONT.tr}). These are available at~\cite{isoquant-graphs}.

\paragraph*{Discussion.} We group the graphs in each dataset by width (column ``$w$''). For each width-interval we show different metrics, namely the number of graphs ``$\#$g'', the average number of arcs ``avg $m$'' (and the maximum number of arcs in parenthesis), the time taken in seconds by the safety preprocessing step in ``prep'', the percentage of fixed arc-variables $x_{uvi}$ of the ILP out of all $mk$ variables in ``vars'', the number of solved instances by the ILP with and without safety in ``$\#$solved'' (and in parenthesis the average time taken in seconds to solve those instances), and in the column ``$\times$'' we show the average speedup obtained by advising the linear program with safety. The speed-ups are computed by considering for each graph the ratio between the time taken by the original ILP (or 300 sec. if this timed-out) divided by the time taken by the optimized ILP, and averaging these ratios. In the ``prep'' and ``vars'' column, a dash is written only if no instance was solved with safety information. In the remaining columns a dash means not applicable.

In \Cref{table:rb-100} we show the experimental results of the safety-optimized ILPs for \ILPRobust.(As explained before, for this problem it is correct to select $C=A$.) When the width of the graph is small, the speed-ups are also small. However, such graphs are fast to solve without safety in the first place. The challenging cases are the graphs of larger widths and more arcs (indeed, the solver with no safety information performs poorly and it becomes slower as the width increases). The speed-up of the solver significantly increases as the width of the graph increases, becoming more than 200$\times$ in some datasets. Moreover, these speed-up are obtained almost for free, in the sense that the time needed to compute safe sequences is negligible for all datasets, i.e.~below 0.1 seconds, and often even much faster. We highlight the notable difference between the number of instances solved with no safety and with safety. In \Cref{table:lq-25} for \ILPLQ we observe the same phenomena, obtaining in some datasets speed-ups of 300$\times$ on graphs of larger widths, and solving always more instances with safety than without safety. In~\Cref{sec:additional-experiments} results for $C \subsetneq A$ in the \ILPRobust and for $C=A$ in the \ILPLQ can be found.

\paragraph*{Acknowledgments.} We are extremely grateful to Andrey Prjibelski for producing the IsoQuant graphs and to Manuel Cáceres for fruitful discussions on safe sequences. This work has received funding from the European Research Council (ERC) under the European Union's Horizon 2020 research and innovation programme (grant agreement No.~851093, SAFEBIO), and from the Research Council of Finland grants No.~346968, 358744. Co-funded by the European Union (ERC, SCALEBIO, 101169716). Views and opinions expressed are however those of the author(s) only and do not necessarily reflect those of the European Union or the European Research Council. Neither the European Union nor the granting authority can be held responsible for them.

\hspace{-1.1cm}\includegraphics[width=6cm]{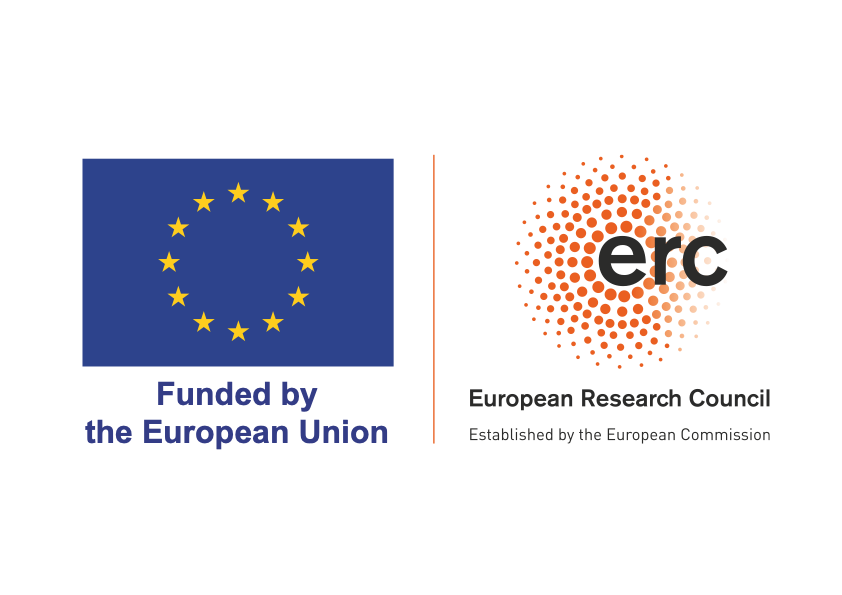}

\begin{table}[p]
\caption{Experimental results of \ILPRobust. Safety is with respect to path covers of $C = A$.\label{table:rb-100}}
\begin{center}
\begin{tabular}{|r|r|r|r|r|r|r|r|r|}
\hline
& \multirow{2}{*}{$w$} & \multirow{2}{*}{\#g} & \multirow{2}{*}{\shortstack{avg $m$\\(max $m$)}} & \multirow{2}{*}{prep (s)} & \multirow{2}{*}{\shortstack{vars \\ (\%)}} & \multicolumn{2}{c|}{\#solved (avg time (s))} & \multirow{2}{*}{$\times$} \\ \cline{7-8}
&      &        &         &      &   & no safety & safety &  \\ \hline

\multirow{3}{*}{\rotatebox{90}{\shortstack{\textbf{Zebrafish}\\\textbf{}}}}& 1-3 & 15405 & 14 (210) & 0.003 & 87.5 & 15405 (0.011) & 15405 (0.011) & 1.2 \\
& 4-6 & 239 & 35 (156) & 0.007 & 29.5 & 239 (1.927) & 239 (0.095) & 14.6 \\
& 7-9 & 6 & 70 (158) & 0.015 & 15.3 & 4 (71.070) & 6 (0.645) & 244.3 \\
& 10+ & 1 & 81 (81) & 0.016 & 13.1 & 0 (-) & 1 (3.421) & 88.1 \\

\hline

\multirow{3}{*}{\rotatebox{90}{\shortstack{\textbf{Human}\\\textbf{}}}}& 1-3 & 10729 & 15 (331) & 0.003 & 78.9 & 10729 (0.020) & 10729 (0.014) & 1.4 \\
& 4-6 & 947 & 35 (163) & 0.008 & 24.7 & 944 (4.031) & 947 (0.147) & 19.0 \\
& 7-9 & 92 & 54 (105) & 0.012 & 14.1 & 63 (63.346) & 92 (3.924) & 116.6 \\
& 10+ & 12 & 80 (104) & 0.019 & 11.7 & 2 (189.233) & 9 (32.264) & 35.0 \\

\hline

\multirow{3}{*}{\rotatebox{90}{\shortstack{\textbf{Mouse}\\\textbf{}}}}& 1-3 & 12280 & 14 (128) & 0.003 & 84.4 & 12280 (0.014) & 12280 (0.011) & 1.3 \\
& 4-6 & 749 & 35 (126) & 0.007 & 27.4 & 749 (2.295) & 749 (0.121) & 15.0 \\
& 7-9 & 60 & 58 (231) & 0.015 & 17.5 & 43 (52.717) & 60 (4.955) & 73.3 \\
& 10+ & 9 & 79 (121) & 0.014 & 15.2 & 1 (262.851) & 7 (34.163) & 69.2 \\

\hline

\multirow{3}{*}{\rotatebox{90}{\shortstack{\textbf{SRR020730}\\\textbf{}}}}& 1-3 & 35069 & 9 (186) & 0.002 & 84.7 & 35069 (0.015) & 35069 (0.010) & 1.4 \\
& 4-6 & 4497 & 33 (171) & 0.007 & 20.4 & 4496 (2.043) & 4497 (0.124) & 12.7 \\
& 7-9 & 1008 & 52 (131) & 0.010 & 11.9 & 844 (42.962) & 1007 (1.692) & 117.1 \\
& 10+ & 296 & 75 (184) & 0.015 & 7.9 & 80 (63.971) & 289 (22.945) & 142.5 \\

\hline

\multirow{3}{*}{\rotatebox{90}{\shortstack{\textbf{Mouse}\\\textbf{ONT}}}}& 1-3 & 18527 & 11 (129) & 0.003 & 81.4 & 18527 (0.013) & 18527 (0.011) & 1.3 \\
& 4-6 & 3083 & 31 (387) & 0.007 & 30.2 & 3083 (0.203) & 3083 (0.041) & 4.8 \\
& 7-9 & 762 & 48 (157) & 0.010 & 18.1 & 755 (2.495) & 762 (0.180) & 38.7 \\
& 10+ & 442 & 87 (589) & 0.021 & 11.1 & 372 (18.194) & 437 (0.611) & 266.0 \\

\hline

\multirow{3}{*}{\rotatebox{90}{\shortstack{\textbf{Mouse}\\\textbf{ONT.tr}}}}& 1-3 & 18029 & 12 (131) & 0.003 & 81.1 & 18029 (0.011) & 18029 (0.011) & 1.1 \\
& 4-6 & 3028 & 32 (187) & 0.007 & 30.6 & 3028 (0.187) & 3028 (0.042) & 3.8 \\
& 7-9 & 803 & 50 (419) & 0.012 & 19.0 & 796 (2.344) & 803 (0.182) & 21.9 \\
& 10+ & 476 & 89 (549) & 0.022 & 11.5 & 403 (17.204) & 472 (0.982) & 162.1 \\

\hline

\multirow{3}{*}{\rotatebox{90}{\shortstack{\textbf{Mouse}\\\textbf{PacBio}}}}& 1-3 & 14256 & 14 (382) & 0.003 & 83.5 & 14256 (0.010) & 14256 (0.011) & 1.2 \\
& 4-6 & 1376 & 33 (187) & 0.006 & 33.3 & 1376 (0.196) & 1376 (0.040) & 5.2 \\
& 7-9 & 182 & 49 (159) & 0.009 & 20.5 & 181 (4.083) & 182 (0.144) & 48.1 \\
& 10+ & 63 & 109 (1178) & 0.026 & 13.3 & 53 (22.188) & 62 (2.557) & 210.3 \\

\hline
\end{tabular}
\end{center}
\end{table}

\begin{table}[p]
\caption{Experimental results of \ILPLQ. Safety is with respect to path covers of the subset $C$ of arcs whose weight is at least as large as the 25-th percentile of the weights of all arcs.
\label{table:lq-25}}
\begin{center}
\begin{tabular}{|r|r|r|r|r|r|r|r|r|}
\hline
& \multirow{2}{*}{$w$} & \multirow{2}{*}{\#g} & \multirow{2}{*}{\shortstack{avg $m$\\(max $m$)}} & \multirow{2}{*}{prep (s)} & \multirow{2}{*}{\shortstack{vars \\ (\%)}} & \multicolumn{2}{c|}{\#solved (avg time (s))} & \multirow{2}{*}{$\times$} \\ \cline{7-8}
&      &        &         &      &   & no safety & safety &  \\ \hline

\hline

\multirow{3}{*}{\rotatebox{90}{\shortstack{\textbf{Zebrafish}\\\textbf{}}}}& 1-3 & 15405 & 14 (210) & 0.003 & 83.1 & 15405 (0.031) & 15405 (0.023) & 1.5 \\
& 4-6 & 239 & 35 (156) & 0.007 & 18.1 & 224 (13.295) & 238 (3.061) & 15.9 \\
& 7-9 & 6 & 70 (158) & 0.011 & 8.4 & 0 (-) & 2 (62.840) & 4.8 \\
& 10+ & 1 & 81 (81) & - & - & 0 (-) & 0 (-) & - \\

\hline

\multirow{3}{*}{\rotatebox{90}{\shortstack{\textbf{Human}\\\textbf{}}}}& 1-3 & 10729 & 15 (331) & 0.003 & 74.6 & 10729 (0.046) & 10729 (0.031) & 1.8 \\
& 4-6 & 947 & 35 (163) & 0.007 & 16.3 & 862 (17.153) & 930 (6.334) & 12.9 \\
& 7-9 & 92 & 54 (105) & 0.011 & 9.6 & 4 (69.195) & 53 (71.207) & 29.1 \\
& 10+ & 12 & 80 (104) & - & - & 0 (-) & 0 (-) & - \\

\hline

\multirow{3}{*}{\rotatebox{90}{\shortstack{\textbf{Mouse}\\\textbf{}}}}& 1-3 & 12280 & 14 (128) & 0.003 & 80.0 & 12280 (0.035) & 12280 (0.026) & 1.6 \\
& 4-6 & 749 & 35 (126) & 0.007 & 16.9 & 696 (14.851) & 744 (5.917) & 12.4 \\
& 7-9 & 60 & 58 (231) & 0.015 & 9.9 & 0 (-) & 16 (64.253) & 121.9 \\
& 10+ & 9 & 79 (121) & - & - & 0 (-) & 0 (-) & - \\

\hline

\multirow{3}{*}{\rotatebox{90}{\shortstack{\textbf{SRR020730}\\\textbf{}}}}& 1-3 & 35069 & 9 (186) & 0.002 & 83.2 & 35069 (0.026) & 35069 (0.017) & 1.7 \\
& 4-6 & 4497 & 33 (171) & 0.007 & 15.5 & 4428 (13.082) & 4494 (1.275) & 16.7 \\
& 7-9 & 1008 & 52 (131) & 0.011 & 8.9 & 269 (127.595) & 899 (25.515) & 82.7 \\
& 10+ & 296 & 75 (184) & 0.015 & 6.8 & 1 (7.363) & 104 (61.880) & 68.0 \\

\hline

\multirow{3}{*}{\rotatebox{90}{\shortstack{\textbf{Mouse}\\\textbf{ONT}}}}& 1-3 & 18527 & 11 (129) & 0.003 & 75.9 & 18527 (0.020) & 18527 (0.017) & 1.6 \\
& 4-6 & 3083 & 31 (387) & 0.007 & 22.5 & 3077 (1.448) & 3082 (0.310) & 7.6 \\
& 7-9 & 762 & 48 (157) & 0.010 & 13.5 & 693 (15.761) & 754 (1.602) & 116.2 \\
& 10+ & 442 & 87 (589) & 0.018 & 8.5 & 222 (37.135) & 405 (11.143) & 301.0 \\

\hline

\multirow{3}{*}{\rotatebox{90}{\shortstack{\textbf{Mouse}\\\textbf{ONT.tr}}}}& 1-3 & 18029 & 12 (131) & 0.003 & 74.7 & 18029 (0.020) & 18029 (0.017) & 1.5 \\
& 4-6 & 3028 & 32 (187) & 0.007 & 21.7 & 3012 (1.693) & 3022 (0.426) & 7.1 \\
& 7-9 & 803 & 50 (419) & 0.011 & 13.5 & 727 (14.958) & 794 (2.087) & 76.2 \\
& 10+ & 476 & 89 (549) & 0.021 & 8.8 & 251 (32.631) & 423 (8.422) & 183.2 \\

\hline

\multirow{3}{*}{\rotatebox{90}{\shortstack{\textbf{Mouse}\\\textbf{PacBio}}}}& 1-3 & 14256 & 14 (382) & 0.003 & 74.4 & 14256 (0.034) & 14256 (0.025) & 1.7 \\
& 4-6 & 1376 & 33 (187) & 0.007 & 22.1 & 1356 (2.622) & 1372 (1.189) & 12.3 \\
& 7-9 & 182 & 49 (159) & 0.010 & 14.0 & 146 (37.428) & 174 (6.627) & 137.8 \\
& 10+ & 63 & 109 (1178) & 0.015 & 10.3 & 15 (110.413) & 46 (17.775) & 370.6 \\

\hline

\end{tabular}
\end{center}
\end{table}

\newpage

\bibliography{main}

\appendix

\newpage

\section{Safe sequences for path covers of a subset of nodes}
\label{sec:subset-path-covers}

In this section we describe how to generalize our results of \Cref{sec:safe-sequences-dominators} for safe sequences with respect to path covers where only some nodes have to be covered. More specifically, given an $s$-$t$ DAG $G=(N,A)$ and a set of nodes $C\subseteq N$, we say that a \emph{$C$-path cover} of $G$ is a set of paths where each node of $C$ is in some path of the cover. In particular, if $C=\emptyset$ then the empty set of paths is a $C$-path cover and moreover no sequence is safe. On the other hand, if $C=N$ then we are in the case of standard path covers and we get \Cref{thm:safe_sequences_char}.

The characterization of safe sequences given in \Cref{thm:safe_sequences_char} easily generalizes to arbitrary sets $C\subseteq N$, and hence dominators still prove to be an appropriate tool to reason about safety on $C$-path covers.

\begin{theorem}[Characterization of safe sequences for subset-path covers]
\label{thm:safe_sequences_char_subset}
    Let $G=(N,A)$ be an $s$-$t$ DAG, $C\subseteq N$ a set of nodes, and $X$ be a sequence of nodes. The following statements are equivalent:
    \begin{enumerate}
        \item $X$ is safe for $C$-path covers.
        \item There is a node $u \in C$ such that every $s$-$t$ path covering $u$ contains $X$.
        \item There is a node $v \in C$ such that $X$ is contained in the sequence obtained as the merge of the sequences of $s$-$v$ and $v$-$t$ cutnodes.
    \end{enumerate}
\end{theorem}

Recall that in the characterization of maximal safe sequences (\Cref{thm:cores-nodes}) we assume to be given a compressed graph. This technicality was mainly introduced for the sake of clarity in characterizing maximal safe sequences. Indeed, if the graph is not compressed the leaf characterization fails. For example, subdividing each arc in the graph of \Cref{fig:subset-cover} results in a graph whose dominator trees have no leaves in common.

We could then expect that compressing the graph while somehow preserving the restrictions imposed by $C$ would allow for a direct analogue of \Cref{thm:cores-nodes} for $C$-path covers. We claim that this approach is not enough to get an optimal time algorithm. Let $P$ be the set of maximal unitary paths of a graph $G$ and let $G'=(N',A')$ be $G$ after compression. Clearly, there is a correspondence between $N'$ and $P$. The set $C' := \{ v \in N' \mid \text{there is a node in $C$ that got compressed into $v$} \}$ is consistent with $C$, i.e., the instances $G,C$ and $G',C'$ have essentially the same set of path covers. Observe now in \Cref{fig:subset-cover} that $C'=C$ since $G$ is already compressed, but no \emph{$C$-leaf}, i.e., a node in $C$ with no descendants in $C$, common to both dominator trees identifies the maximal safe sequence $s,a,f,t$ via paths to the root in the dominator trees. The key observation is that the situation depicted by the paths $af$ and $fa$ in the dominator trees of \Cref{fig:subset-cover} is analogous to that of paths in the $s$- and $t$-dominator trees that correspond to nontrivial unitary paths in $G$.

Say that a path in a tree is \emph{univocal} if each of its nodes but the deepest one has exactly one child; this definition captures exactly the paths that arose in the discussion above. If $p$ is a univocal path in one dominator tree and there is a univocal path in the other dominator tree with the same set of nodes, then the sequence of nodes in that path must appear in reversed order with respect to $p$, since $G$ is acyclic. Thus, for ease of writing, we fix the orientation given by the $s$-dominator tree to syntactically describe univocal paths appearing in both trees (e.g., $af$ is a univocal path common to both trees in \Cref{fig:subset-cover}).
Note that identifying all the univocal paths in the dominator trees can be done in linear time and that the extension of any two nodes in a univocal path gives the same sequence. Let us build the dominator trees of $G$, where we add abstract blue arcs representing the dominance relation restricted to $C$. Let us call this tree by the \emph{blue dominator tree}; essentially, this tree encodes the domination-relation restricted to $C$. With this construct we can now argue about the maximality of safe sequences for $C$-path covers, similarly to \Cref{thm:cores-nodes}.

\begin{figure}[t]
    \centering
    \includegraphics[width=1\linewidth]{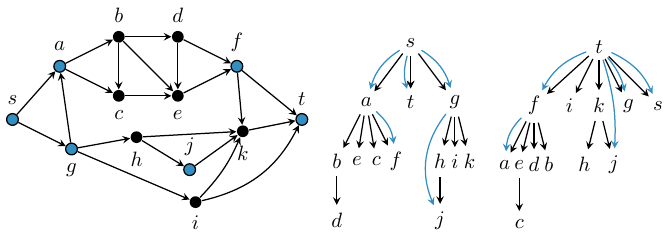}
    \caption{An $s$-$t$ DAG with the nodes belonging to $C$ in blue ($s$ and $t$ are included without loss of generality). On its right, the dominator trees of $G$ are drawn with additional blue arcs to indicate the immediate-dominator relation restricted to $C$. The only blue node not having blue descendants in both trees is the node $j$. Nodes $a$ and $f$ identify the same sequence via extensions. The path $af$ is univocal in the blue $s$-dominator tree and has no descendants in $C$, and the same in the blue $t$-dominator tree. The path $gj$ is univocal only in the blue $s$-dominator tree, the path $saf$ is not since $s$ has more than one blue children.}
    \label{fig:subset-cover}
\end{figure}

Let $u \in C$ be a node. We aim to show that $\ext{u}$ is a maximal safe sequence for $C$-path covers if and only if there is a univocal blue path appearing in both dominator trees that contains $u$ where its deepest node in the $s$-dominator tree as well as its deepest node in the $t$-dominator tree has no descendants in $C$. First, if a blue univocal path appears in both trees and its deepest nodes do not have descendants in $C$ in either tree, then the extension of any of its nodes is clearly maximal. Note that $\ext{.}$ is still computed over the dominator trees of $G$, this discussion serves the purpose of understanding which nodes should we extend.
To see the other direction we argue by contradiction. Suppose that $\ext{u}$ is maximally safe but there is no univocal blue path containing $u$ that appears in both dominator trees and whose deepest nodes have no descendants in $C$.

First, observe that $u$ cannot be a $C$-leaf in both trees, otherwise $u$ itself would be a univocal blue path trivially containing itself and without any descendants in $C$ in either tree.
Therefore, and without loss of generality, there is a node $v \in C$ such that $uv$ is a blue arc in the $s$-dominator tree. Let $k \in \mathbb{N}^+$ be such that $u$ is the $k$-th ancestor of $v$ in the $s$-dominator tree.
If there is a node $v' \in C$ distinct from $v$ such that $uv'$ is a blue arc, then it cannot be the case that both $v$ and $v'$ $t$-dominate $u$. To see why, suppose that they both $t$-dominate $u$ and say that $v$ is the closest to $t$ in the $t$-dominator tree. Then we can take a $u$-$v$ path avoiding $v'$ and extend it to $t$ also avoiding $v$', thus contradicting the fact that $v'$ $t$-dominates $u$. Therefore, we can apply \Cref{lem:st-dom-immediate-relation} to $u$ and $v$ (or $u$ and $v'$, but say to $v$ without loss of generality) to conclude that $w=dom_t(u,k)$ \emph{strictly} $t$-dominates $v$. We now want to show that the extension of $v$ properly contains the extension of $u$. It suffices to show that no node in $C$ exists between $w$ and $u$ in the $t$-dominator tree, since we can then apply the same reasoning described in the proof of \Cref{thm:cores-arcs} to conclude that $\ext{v}$ properly contains $\ext{u}$. Indeed, if there is a node $x \in C$ between $w$ and $u$ in the $t$-dominator tree, then $x$ must be between $u$ and $v$ in the $s$-dominator tree, for otherwise we can take a $u$-$v$ path avoiding $x$ and prolong it with a $v$-$t$ path also avoiding $x$, which results in a $u$-$t$ path avoiding $x$ and contradicts the fact that $x$ $t$-dominates $u$ (see \Cref{fig:proof-sketch}).

It remains to analyze the case where $uv$ is a univocal path in both dominator trees. By our assumption, in at least one of the dominator trees there is a descendant in $C$ of $v$ (resp. $u$) in the $s$-dominator tree (resp. $t$-dominator tree). Suppose there is a node $x \in C$ such that $uvx$ is a univocal path in the $s$-dominator tree. If $xvu$ is also univocal in the $t$-dominator tree, then we can repeat this process until we obtain a maximal univocal blue path $p$ appearing in both dominator trees that contains $u$. By our assumption, the deepest node of $p$ in at least one of the trees has a descendant in $C$. Without loss of generality, suppose this node is in the $s$-dominator tree and call it $z$. Let $k \in \mathbb{N}^+$ be such that $u$ is the $k$-th ancestor of $z$ in the $s$-dominator tree ($k \geq 2$). Now we apply \Cref{lem:st-dom-immediate-relation} to $u$ and $z$ and conclude that $w=dom_t(u,k)$ $t$-dominates $z$. By the maximality of $p$, it cannot be the case that $w = z$, $z$ has a unique blue child in the $t$-dominator tree, and the deepest node of $p$ in the $s$-dominator tree has $z$ as its unique blue child, otherwise adding $z$ to $p$ would yield a univocal path common to both trees that properly contains $p$. It is enough to observe that if $w=z$ and $z$ has a unique blue child in the $t$-dominator tree, then the deepest node of $p$ in the $s$-dominator tree has a blue child other than $z$, say $z'$. Then, we are in the same case as in the previous paragraph and we can conclude that $\ext{z'}$ properly contains $\ext{u}$ (or the extension of any node in $p$, as they all yield the same extension). 
These are all the possible cases and the argument is finished.

We are now ready to prove the main theorem concerning safe sequences in $C$-path covers.

\begin{theorem}[Maximal safe sequence enumeration for $C$-path covers]
\label{thm:representation-optimal-enumeration-C}
    Let $G$ be an $s$-$t$ DAG with $n$ nodes and $m$ arcs, and let $C \subseteq N$. The following hold.
    \begin{enumerate}
        \item There is an $O(m+o)$ time algorithm outputting the set of all maximal safe sequences for $C$-path covers with no duplicates, where $o$ denotes the total length of all the maximal safe sequences.
        
        \item There is an $O(n)$-size representation of all the maximal safe sequences of $G$ for $C$-path covers that can be built in $O(m+n)$ time.
    \end{enumerate}
\end{theorem}

\begin{proof}
    
    \textbf{1:} Build the dominator trees of $G$. Next, build the blue dominator trees of $G$ with respect to $C$. With respect to the blue $s$-dominator tree, define $P_s$ as the set of all univocal paths whose deepest nodes do not have any descendant in $C$. Define $P_t$ analogously. The set $P_s \cap P_t$ thus contains the maximal univocal paths with respect to \emph{both} blue dominator trees, i.e., adding any node to any path in this set, makes the path not univocal in some blue dominator tree. Pick an arbitrary node from a path in $P_s \cap P_t$ and output its extension. It follows from the discussion above that this extension results in a maximal safe sequence. Moreover, the fact that no maximal safe sequence is reported twice follows from the maximality of the paths in $P_s \cap P_t$. This procedure runs in time equal to the length of all the maximal safe sequences plus the time required to build two dominator trees and to compute the blue trees, i.e., $O(m+o)$.

    \textbf{2:} From point 1 it follows that the extension of a node $v$ belonging to a path in $P_s \cap P_t$ is a maximal safe sequence for $C$-path covers, as well as the converse. Thus, the dominator trees of $G$ together with the blue arcs and the set $P_s \cap P_t$ encode the necessary information to represent all the maximal safe sequences. All these computations can be performed in $O(m+n)$ time, and the dominator trees together with the blue arcs and the set $P_s \cap P_t$ require only $O(n)$ space.\hfill $\square$
\end{proof}

Conceptually, for a graph $G$ and set $C$, we can consider the $s$-$t$ DAG $H_G=(C \cup \{s,t\},A)$ where $A := \{ uv \mid u,v\in C \; \wedge \text{ there is a $u$-$v$ path in $G$ without nodes in $C$ but $u$ and $v$}\}$. If we now consider $G$ to be the graph in \Cref{fig:subset-cover}, we see that the dominator trees of $H_G$ are exactly the blue trees shown on top of the dominator trees of $G$. Now we can prove a direct analogue of \Cref{thm:cores-nodes} for $C$-path covers, as we can compress $H_G$ into $H'_G$ and apply \Cref{lem:compressed-graph}. Then, to find the maximal safe sequences, we simply output $\ext{u}$ in the dominator trees of $G$ for every node $u$ in the node set of $H'_G$ such that $u$ is a common leaf in the dominator trees $H'_G$. However, for algorithms, we cannot afford to build $H_G$ as it involves the computation of the reachability relation of a graph.

\section{Safe sequences in arc-path covers}\label{sec:arc-path-covers}

In this section we present analogous results to those in \Cref{sec:safe-sequences-dominators} for safe sequences of arcs in path covers of the arcs. We begin by redefining some concepts.

We work on $s$-$t$ DAGs with node set $N$ and arc set $A$ and allow parallel arcs, so that in fact are working on directed acyclic multigraphs. For $X \subseteq A$ we denote a \emph{sequence of arcs} $X$ through arcs $a_1,a_2,\dots,a_k$ as $X=a_1,\dots,a_k$ if each arc in the sequence reaches the one after it; a path is now denoted as a contiguous sequence of arcs. If $Y$ is a sequence of arcs and each arc of $X$ is contained in $Y$ we say that $X$ is a \emph{subsequence} of $Y$, or that $X$ is \emph{contained} in $Y$. Two paths are distinct if they differ in at least one arc, and two path cover are distinct if they differ in at least one path. A sequence of arcs is safe for a DAG $G$ if it is a subsequence of some $s$-$t$ path in every path cover of $G$.

We say that $xy$ is a \emph{$u$-$v$ bridge} if there is no $u$-$v$ path in $G$ without $xy$; equivalently, all $u$-$v$ paths in $G$ contain the arc $xy$. The concepts introduced for dominators in~\Cref{sec:notation-preliminaries} can be defined for arcs: given an arc $uv$, we say that $xy$ $s$-dominates $uv$ iff every $s$-$u$ path contains $xy$, and $zw$ $t$-dominates $uv$ iff every $v$-$t$ path contains $zw$. We can also refine the notion of extension for arcs by saying that $\ext{uv}$ consists of the sequence of $s$-$u$ bridges concatenated with $uv$, followed by the sequence of $v$-$t$ bridges.

A \emph{unitary path} with respect to arc-path covers is a path where each node that is not the first nor the last has indegree and outdegree equal to one. The \emph{arc-compressed graph} of $G$ is a graph where each maximal unitary path is replaced by an arc from the first node to the last node of the path. Note that in an arc-compressed graph, no node has both indegree and outdegree equal to one. Note also that we might introduce parallel arcs during compression if a transitive arc exists from the first to the last node of the unitary path, so indeed it is important that we allow multigraphs to begin with. Arc-compressing a graph can be done in linear time, similarly to point 1 of \Cref{lem:compressed-graph}.
The \emph{line graph} of a directed graph $G=(V,E)$, denoted by $L(G)$, is a graph where each node corresponds to an arc of $G$ and where there is an arc from node $uv$ to node $xy$ if $v=x$; essentially, arcs of $L(G)$ encode paths of two arcs in $G$.

With this, we can already characterize safe sequences of arcs.

\begin{theorem}
\label{thm:safe_sequences_char_arc}
    Let $G=(N,A)$ be a DAG and $X$ a sequence of arcs. The following statements are equivalent:
    \begin{enumerate}
        \item $X$ is safe for path covers.
        \item There is an arc $uv \in A$ such that every $s$-$t$ path covering $uv$ contains $X$.
        \item There is an arc $xy \in A$ such that $X$ is contained in the sequence resulting from the concatenation of the sequences of $s$-$x$ and $y$-$t$ bridges.
    \end{enumerate}
\end{theorem}
\begin{proof}
    Apply \Cref{thm:safe_sequences_char} to $L(G)$.\hfill $\square$
\end{proof}

\begin{lemma}[Characterization of arc-dominators]
\label{lemma:arc_dominator_char}
    Let $G=(N,A)$ be an $s$-$t$ DAG. An arc $uv \in A$ is an $s$-dominator (resp. $t$-dominator) if and only if $d_v^+ \geq 1$ and $d_v^-= 1$ (resp. $d_u^- \geq 1$ and $d_u^+ = 1$).
    
\end{lemma}
\begin{proof}
    Without loss of generality we prove the statement for $s$-dominators.
    
    $(\Leftarrow)$ If the unique in-neighbor of $v$ is $u$ and there exists an arc $vw$, then $uv$ is the immediate dominator of $vw$. In fact, $uv$ is the immediate dominator of every outgoing arc from $v$.
    
    $(\Rightarrow)$ Suppose for a contradiction that $d_v^+ < 1$ or $d_v^- \neq 1$. If $d_v^+ = 0$ then $v=t$ and so $uv$ is not a dominator, so suppose $v$ has at least one out-neighbor. If $v$ has an in-neighbor $z \neq u$ then there is an $s$-$z$ path avoiding $uv$, for otherwise there would be a cycle in $G$ through $v$. Therefore $uv$ does not dominate any of the outgoing arcs from $v$, and thus it does not dominate any of the arcs it reaches: to any path containing $uv$ we change the prefix of the path until $v$ with an $s$-$v$ path through $z$ avoiding $uv$.\hfill $\square$
\end{proof}

We can also get the arc-version analogue of \Cref{lemma:no_redundancy_nodes} by considering the line graph. However, a stronger property holds for arc-compressed graphs, which we prove next.

\begin{lemma}
\label{lemma:no_redundancy_arcs}
    Let $G=(N,A)$ be an arc compressed $s$-DAG. Then every $s$-dominator $uv \in A$ immediately dominates at least two of the outgoing arcs from $v$.
\end{lemma}
\begin{proof}
    Let $uv$ be an $s$-dominator. By \Cref{lemma:arc_dominator_char} we have $d^-_v=1$, and so $uv$ immediately dominates each outgoing arc from $v$. Since $G$ is compressed we must have $d_u^+ \geq 2$, proving the statement.\hfill $\square$
\end{proof}

\begin{lemma}
\label{lemma:atm1_child_arcs}
    Let $G=(N,A)$ be an $s$-$t$ DAG and let $uv,xy \in A$ be arcs. If $uv$ immediately $s$-dominates $xy$, then the immediate $t$-dominator of $uv$ $t$-dominates $xy$.
\end{lemma}
\begin{proof}
    Let $a$ denote the node corresponding to $uv$ in $L(G)$ and analogously $b$ to $xy$.
    Apply \Cref{lem:st-dom-immediate-relation} to $L(G)$, $a$, $b$, and $k=1$.\hfill $\square$
\end{proof}

In a compressed graph, whenever $uv$ immediately $s$-dominates $xy$ and $xy$ immediately $t$-dominates $uv$, then $v \neq x$. Otherwise, by \Cref{lemma:arc_dominator_char}, we would have $d^-_v = 1$ because $uv$ is an $s$-dominator and $d^+_v = 1$ because $xy$ is a $t$-dominator, which contradicts the assumption that the graph is compressed.

\begin{theorem}[Characterization of maximal safe sequences]
\label{thm:cores-arcs}
    Let $G=(N,A)$ be an arc-compressed $s$-$t$ DAG and let $uv \in A$ be an arc. Then $\ext{uv}$ is a maximal safe sequence of arcs if and only if $uv$ is a leaf in both the $s$- and $t$-dominator trees of G.
\end{theorem}
\begin{proof}
    $(\Leftarrow)$ If $uv$ is a leaf in both dominator trees then the sequence that $uv$ identifies via paths to the roots of the dominator trees is maximal: no other arc can be added to this sequence by the equivalence of safe sequences and paths in the dominator trees given by \Cref{thm:safe_sequences_char_arc}.

    $(\Rightarrow)$ Suppose that $uv$ is not a leaf in the dominator tree of $s$. We apply \Cref{lemma:no_redundancy_arcs} to $G$ to get two arcs $s$-dominated by $uv$, $xy$ and $zw$. It suffices to show that the lowest common ancestor in the $t$-dominator tree between $uv$ and $xy$ or between $uv$ and $zw$ is the immediate $t$-dominator of $uv$. By \Cref{lemma:atm1_child_arcs}, we can assume that $xy$ is not the immediate $t$-dominator of $uv$. Let $a$ denote the lowest common ancestor in the $t$-dominator tree between $uv$ and $xy$. If $a$ is not the immediate $t$-dominator of $uv$ then it is an ancestor of it. Consider an $xy$-$t$ path avoiding the immediate dominator of $uv$ and consider a $uv$-$xy$ path avoiding the immediate $t$-dominator of $uv$; these paths together form a $uv$-$t$ path avoiding the immediate $t$-dominator of $uv$, a contradiction.\hfill $\square$
\end{proof}

Since we cannot get a direct analogues of \Cref{thm:mn-all-only-maximal} and \Cref{thm:representation-optimal-enumeration} via the line graph construction, we now present two theorems concerning arc-path covers whose proofs are completely analogous to \Cref{thm:mn-all-only-maximal} and \Cref{thm:representation-optimal-enumeration}. Essentially, we arc-compress the graph and then build its dominator trees with respect to the arc-dominance relation.

\begin{theorem}
\label{thm:mn-all-only-maximal-arcs}
    Given an $s$-$t$ DAG $G$ with $n$ nodes and $m$ arcs, we can compute all and only the maximal safe sequences of arcs of $G$ in $O(m^2)$-time, without constructing dominator trees.
\end{theorem}

\begin{theorem}[Representation and optimal enumeration of safe sequences of arcs]
\label{thm:representation-optimal-enumeration-arcs}
    Let $G$ be an $s$-$t$ DAG with $n$ nodes and $m$ arcs. The following hold.
    \begin{enumerate}
        \item There is an $O(m+o)$ time algorithm outputting the set of all maximal safe sequences of arcs with no duplicates, where $o$ denotes the total length of all the maximal safe sequences of arcs.
        
        \item There is a $O(m)$-size representation of all the maximal safe sequences of arcs of $G$ that can be built in $O(m+n)$ time.
    \end{enumerate}
\end{theorem}

In \Cref{sec:subset-path-covers} we showed how to generalize path covers of all the nodes of the graph to subsets of nodes and ultimately derive a generalization of \Cref{thm:representation-optimal-enumeration}, which materialized in \Cref{thm:representation-optimal-enumeration-C}. The same generalization can easily be done for arcs, analogously to how were able to derive \Cref{thm:representation-optimal-enumeration-arcs} from the results in the main text with minor adaptations.

\section{Omitted proofs}
\label{sec:additional-proofs}

\dominatorscharacterization*

\begin{proof}
    We prove the statement only for $s$-dominators, since the statement for $t$-dominators follows by symmetry.
    
    $(\Leftarrow)$ If the unique in-neighbor of $v$ is $u$, then every $s$-$v$ path contains $u$, and moreover $u$ is always the node immediately preceding $v$ in any of those paths, and thus $u$ must be the immediate dominator of $v$. In fact, for the same reasons, $u$ is the immediate dominator of every node whose unique in-neighbor is $u$.
    
    $(\Rightarrow)$ Suppose $u$ is a dominator, and let $N^+_u = \{v_1,\dots,v_k\}$. If each $v_i$ has an in-neighbor $z_i \notin \{u\} \cup N^+_u$, then as Ochranov{\'a}~\cite{ochranova1983finding} observes, for every node $w$ dominated by $u$ we can construct an $s$-$w$ path passing through some $z_i$ and thus avoiding $u$.
    
    However, in general this property does not hold, because all in-neighbors of some $v_i$ different from $u$ can belong only to $N^+_u$. The remainder of the proof covers this case by showing that nonetheless we can still build $s$-$w$ paths avoiding $u$.
    
    Let $V$ be the set of nodes on a $v_i$-$v_j$ path for $i, j \in \{1,\dots,k\}$ ($i$ possibly equal to $j$), and let $G_V$ be the subgraph induced by $V$. Since $G$ is acyclic, also $G_V$ is acyclic. Let $v_{i_1}, \dots, v_{i_t}$ be the sources of $G_V$. By the assumption, every $v_{i_j}$ has an in-neighbor $z_{i_j}\neq u$. By the definition of $V$, $z_{i_j}$ is not reached by any $v_i \in N^+_u$, since otherwise it would be on a $v_i$-$v_{i_j}$ path (as $z_{i_j}$ reaches $v_{i_j}$), and thus it would belong to $G_V$. This would contradict the fact that $v_{i_j}$ is a source of $G_V$. Thus, no $z_{i_j}$ is reached by $u$. Also by acyclicity, $u$ does not belong to $G_V$.

    Let $w$ be any node that is reachable from $u$. We prove that there is an $s$-$w$ path not using $u$, thus showing that $u$ is not a dominator. Since $u$ reaches $w$, it does so through some out-neighbor, say $v$. Let $v_{i_j}$ be a source of $G_V$ reaching $v$. Consider now an $s$-$z_{i_j}$ path that does not contain $u$. This path exists, since $z_{i_j}$ is not reached by $u$, as proved above. We extend this to $v_{i_j}$, then to $v$ in $G_V$ (without using $u$, since $u$ does not belong to $G_V$), then to $w$. This path does not contain $u$.\hfill $\square$
\end{proof}

\mnallonlymaximal*

\begin{proof}
    First compress $G$ into $G'=(N,A)$ in time $O(m+n)$ (\Cref{lem:compressed-graph}).
    Let $N' \subseteq N$ be the set of nodes $v$ such that there is no node $u \in N$ with $N^+_u = \{v\}$, and such that there is no node $w \in N$ with $N^-_w = \{v\}$.

    Let $v \in N'$. By \Cref{lemma:node_dominator}, we have that $v$ is a leaf in both dominator trees. Since $G'$ is compressed, by \Cref{thm:cores-nodes}, $\ext{v}$ is a maximal safe sequence and every maximal safe sequence is the extension of some node in $N'$. Further, $\ext{v}$, can be computed $O(m)$ time with the algorithm of Cairo et al.~\cite{CAIRO2021103}. Finally, since every node appears exactly once in each dominator tree and extensions from different leaves common to both trees are distinct as they differ in at least the common-leaf, no two nodes in $N'$ identify the same maximal safe sequence via extensions.\hfill $\square$
\end{proof}

The following result is not necessary to get our linear time algorithm (\Cref{thm:representation-optimal-enumeration}). We include it as it complements \Cref{lem:st-dom-immediate-relation} and hence may be of interest.

\begin{lemma}
\label{lem:dom_no_dom}
Let $G=(N,A)$ be an $s$-$t$ directed graph and let $u,v\in N$ be nodes. Suppose that $u$ $s$-dominates $v$ and that $v$ does not $t$-dominate $u$. Then every node $t$-dominated by $v$ is strictly $s$-dominated by $u$.
\end{lemma}
\begin{proof}
    Suppose that there is a node $w$ $t$-dominated by $v$ that is not $s$-dominated by $u$. Pick a $w$-$v$ path $p$ avoiding $u$, which exists since $u$ is not $t$-dominated by $v$. Since $u$ does not $s$-dominate $w$ there is an $s$-$w$ path avoiding $u$, which together with $p$ make a $s$-$v$ path not containing $u$, contradicting the fact that $u$ is an $s$-dominator of $v$.\hfill $\square$
\end{proof}

\section{Applying safe sequences to simplify path-covering ILPs}
\label{sec:encode-safety}

In this section we describe our application of safe sequences to simplify and consequently speed up path-covering ILPs.

\paragraph{Polynomial-size path-covering ILPs.} Dias et al.~\cite{dias2022fast}, and Sashittal et al.~\cite{jumper} showed that path-covering problems admit polynomially-sized ILPs, improving on previous ILP formulations~\cite{cidane,class2,cliiq,isoinfer,isolasso,multitrans,nsmap,ssp,translig} which are potentially exponentially-sized. These work in three parts: (i) $k$ paths are encoded via suitable variables and constraints; (ii) additional constraints (and variables) are added to encode when these paths form a solution; (iii) if necessary, these variables are then used in the formulation of the objective function of the ILP (e.g.~minimization or maximization). Such an ILP can also have a wrapper in which one searches for different numbers $k$ of paths, for each $k$ solves the above-mentioned ILP, and then chooses that value of $k$ optimizing some quantity. In what follows, we review two specific ILPs, and we refer the reader to the original papers~\cite{dias2022fast,jumper} for further details.

For part (i), for each solution path $i$ that needs to be encoded ($i \in \{1,\dots,k\}$), both~\cite{dias2022fast,jumper} introduce a binary variable $x_{uvi}$ associated to each arc $uv$ of the DAG, with the interpretation that $x_{uvi}$ equals 1 if and only if the arc $uv$ belongs to the $i$-th solution path. Since the input graph is a DAG, this can be guaranteed by stating that the sum of all $x_{sui}$ variables on arcs $su$ exiting the source $s$ is exactly 1, and for all nodes $v$ (different from $s$ and $t$) the sum of the binary variables $x_{uvi}$ of arcs $uv$ entering $v$ is equal to the sum of the binary variables $x_{vwi}$ of arcs $vw$ exiting $v$, see~\cite{dias2022fast}.

\paragraph{Fixing ILP variables via safe sequences.} In this paper, we apply safe sequences to fix some $x_{uvi}$ variables to $1$, thus reducing the search space of the ILP solver, using the approach of Grigorjew et al.~\cite{acceleratingILP}. Let $X_1,\dots,X_t$ be sequences in the input DAG such that: (a) each $X_i$, $i \in \{1,\dots,t\}$, must appear in some solution path of the ILP, and (b) there exists no path in the DAG containing two distinct $X_i$ and $X_j$. In the case that every $X_i$ is a path, Grigorjew et al.~\cite{acceleratingILP} noticed that if the two mentioned properties hold, then $X_1,\dots,X_t$ must appear in $t$ distinct paths of among the $k$ solution paths of the ILP. As such, one can, without loss of generality, assign each $X_i$ to the $i$-th solution path. That is, for each $i \in \{1,\dots,t\}$, we set $x_{uvi} = 1$ for each arc $uv$ of $X_i$. 

Finally, if it holds that the ILP solution paths cover all arcs of the graph or all arcs in the set $C$ with respect to which we compute safe sequences, then we can use safe sequences in the above procedure to ensure that the ILP still has the same set of optimal solutions. To find suitable sets of sequences, and to have a large number of binary variables set to 1 in this manner, we follow the approach of~\cite{acceleratingILP}. Namely, for every arc $uv$, we record the longest safe sequence containing $uv$, and use its length as the weight of $uv$. Then, we find a maximum-weight antichain of arcs in the DAG, namely a set of arcs $a_i,\dots,a_t$ such that from no $a_i$ we can reach some other $a_j$, $i,j \in \{1,\dots,t\}$, and its total weight is maximum among all antichains, using the min-flow reduction of~\cite{rival2012graphs}. Then, for each $a_i$ in the maximum-weight antichain, we take the recorded longest safe sequence $X_i$. To see that there can be no path in the DAG containing distinct $X_i$ and $X_j$, notice that such path would contain both the corresponding arcs $a_i$ and $a_j$, which contradicts the fact that $a_i$ and $a_j$ belong to an antichain.

\paragraph{Two path-covering ILPs.} Here we describe two path-covering ILPs that we chose to experiment with. For both, we assume that for every arc $uv$ of the graph we have an associated positive weight $w(uv)$. We also assume part (i) discussed above, namely, the $x_{uvi}$ binary variables, and the constraints discussed above to ensure they induce paths. The overall goal is to find the ``best'' $k$ such paths, and their associated weights $w_1,\dots,w_k$. 

We selected a classical least-squares model (which we call \ILPLQ) that is at the core of several RNA assembly tools e.g.~\cite{isolasso,cidane,ryuto,slide,ireckon,traph}. This minimizes the sum of squared differences between the weight of an arc and the weights of the solution paths passing through the arc. The ILP formulation requires no additional constraints, and the ILP objective function is:\footnote{Note that the terms $x_{uvi}w_i$ are not linear, but can be linearized using a simple technique, see~\cite{dias2022fast,dias2024robust}.}

\begin{equation}
\min \sum_{uv \in A} \left(w(uv) - \sum_{i = 1}^{k}x_{uvi}w_i\right)^2.
\label{eq:lq}    
\end{equation}

While this model is popular, it has quadratic terms in its objective function, which makes it harder to solve. Therefore, as second model (which we call \ILPRobust), we chose one that was recently introduced in~\cite{dias2024robust} and shown to be more accurate than \ILPLQ on the datasets considered therein. In this model, errors are accounted not at the level of individual arcs, but at the level of paths. Namely, for every solution path $P_i$ we have an associated error (or slack) $\rho_i$ which intuitively corresponds to the maximum change in coverage across it. Then, for $k$ paths to be a solution to the ILP, we require that the absolute difference between the weight of an arc and the weights of the solution paths passing through the arc must be below the sum of the errors (or slacks) of the paths passing through the arc. Formally, we add the constraint:\footnote{Note again that the terms $x_{uvi}w_i$ and $x_{uvi}\rho_i$ are not linear, but can be linearized using a simple technique, see~\cite{dias2024robust}.}

\begin{equation}
\left|w(uv) - \sum_{i = 1}^{k}x_{uvi}w_i\right| \leq \sum_{i = 1}^{k}x_{uvi}\rho_i, ~~\forall uv \in A,
\label{eq:robust}
\end{equation}
with objective function:
$\min \sum_{i = 1}^{k}\rho_i$.

\section{Additional experimental results}
\label{sec:additional-experiments}

~~

\begin{table}[h!]
\caption{Experimental results of \ILPRobust. Safety is with respect to path covers of the subset $C$ of arcs whose weight is at least as large as the 25-th percentile of the weights of all arcs. \label{table:rb-25}}
\begin{center}
\begin{tabular}{|r|r|r|r|r|r|r|r|r|}
\hline
& \multirow{2}{*}{$w$} & \multirow{2}{*}{\#g} & \multirow{2}{*}{\shortstack{avg $m$\\(max $m$)}} & \multirow{2}{*}{prep (s)} & \multirow{2}{*}{\shortstack{vars \\ (\%)}} & \multicolumn{2}{c|}{\#solved (avg time (s))} & \multirow{2}{*}{$\times$} \\ \cline{7-8}
&      &        &         &      &   & no safety & safety &  \\ \hline

\multirow{3}{*}{\rotatebox{90}{\shortstack{\textbf{Zebrafish}\\\textbf{}}}}& 1-3 & 15405 & 14 (210) & 0.003 & 83.1 & 15405 (0.012) & 15405 (0.012) & 1.1 \\
& 4-6 & 239 & 35 (156) & 0.006 & 18.1 & 239 (1.968) & 239 (0.374) & 6.1 \\
& 7-9 & 6 & 70 (158) & 0.014 & 9.3 & 4 (72.534) & 6 (23.842) & 33.2 \\
& 10+ & 1 & 81 (81) & - & - & 0 (-) & 0 (-) & - \\

\hline

\multirow{3}{*}{\rotatebox{90}{\shortstack{\textbf{Human}\\\textbf{}}}}& 1-3 & 10729 & 15 (331) & 0.003 & 74.6 & 10729 (0.017) & 10729 (0.013) & 1.3 \\
& 4-6 & 947 & 35 (163) & 0.006 & 16.2 & 944 (3.266) & 947 (0.509) & 9.1 \\
& 7-9 & 92 & 54 (105) & 0.009 & 8.9 & 66 (61.137) & 90 (24.813) & 47.1 \\
& 10+ & 12 & 80 (104) & 0.013 & 7.9 & 1 (83.302) & 5 (14.412) & 21.9 \\

\hline

\multirow{3}{*}{\rotatebox{90}{\shortstack{\textbf{Mouse}\\\textbf{}}}}& 1-3 & 12280 & 14 (128) & 0.003 & 80.0 & 12280 (0.013) & 12280 (0.012) & 1.2 \\
& 4-6 & 749 & 35 (126) & 0.006 & 16.9 & 749 (2.294) & 749 (0.439) & 7.5 \\
& 7-9 & 60 & 58 (231) & 0.011 & 9.5 & 43 (51.997) & 59 (33.739) & 22.1 \\
& 10+ & 9 & 79 (121) & 0.010 & 9.4 & 1 (273.342) & 5 (72.802) & 10.3 \\

\hline

\multirow{3}{*}{\rotatebox{90}{\shortstack{\textbf{SRR020730}\\\textbf{}}}}& 1-3 & 35069 & 9 (186) & 0.002 & 83.2 & 35069 (0.015) & 35069 (0.010) & 1.3 \\
& 4-6 & 4497 & 33 (171) & 0.006 & 15.5 & 4496 (2.041) & 4497 (0.220) & 9.5 \\
& 7-9 & 1008 & 52 (131) & 0.010 & 8.7 & 846 (43.891) & 1008 (5.282) & 67.2 \\
& 10+ & 296 & 75 (184) & 0.015 & 5.9 & 80 (65.021) & 258 (26.940) & 80.6 \\

\hline

\multirow{3}{*}{\rotatebox{90}{\shortstack{\textbf{Mouse}\\\textbf{ONT}}}}& 1-3 & 18527 & 11 (129) & 0.002 & 75.9 & 18527 (0.011) & 18527 (0.010) & 1.2 \\
& 4-6 & 3083 & 31 (387) & 0.006 & 22.4 & 3083 (0.166) & 3083 (0.060) & 3.5 \\
& 7-9 & 762 & 48 (157) & 0.008 & 13.5 & 755 (2.019) & 762 (0.505) & 21.3 \\
& 10+ & 442 & 87 (589) & 0.019 & 8.3 & 374 (16.483) & 437 (0.963) & 145.0 \\

\hline

\multirow{3}{*}{\rotatebox{90}{\shortstack{\textbf{Mouse}\\\textbf{ONT.tr}}}}& 1-3 & 18029 & 12 (131) & 0.002 & 74.7 & 18029 (0.009) & 18029 (0.010) & 1.1 \\
& 4-6 & 3028 & 32 (187) & 0.006 & 21.7 & 3028 (0.154) & 3028 (0.056) & 3.0 \\
& 7-9 & 803 & 50 (419) & 0.009 & 13.5 & 796 (1.999) & 803 (0.351) & 13.1 \\
& 10+ & 476 & 89 (549) & 0.020 & 8.6 & 405 (15.521) & 470 (2.315) & 79.9 \\

\hline

\multirow{3}{*}{\rotatebox{90}{\shortstack{\textbf{Mouse}\\\textbf{PacBio}}}}& 1-3 & 14256 & 14 (382) & 0.003 & 74.4 & 14256 (0.010) & 14256 (0.011) & 1.1 \\
& 4-6 & 1376 & 33 (187) & 0.006 & 22.0 & 1376 (0.203) & 1376 (0.068) & 3.8 \\
& 7-9 & 182 & 49 (159) & 0.010 & 13.9 & 181 (4.180) & 182 (0.524) & 18.8 \\
& 10+ & 63 & 109 (1178) & 0.027 & 9.3 & 53 (22.704) & 62 (3.372) & 76.2 \\

\hline

\end{tabular}
\end{center}
\end{table}

\begin{table}[]
\caption{Experimental results of \ILPLQ. Safety is computed with respect to path covers of the complete set of arcs. \label{table:lq-100}}
\begin{center}
\begin{tabular}{|r|r|r|r|r|r|r|r|r|}
\hline
& \multirow{2}{*}{$w$} & \multirow{2}{*}{\#g} & \multirow{2}{*}{\shortstack{avg $m$\\(max $m$)}} & \multirow{2}{*}{prep (s)} & \multirow{2}{*}{\shortstack{vars \\ (\%)}} & \multicolumn{2}{c|}{\#solved (avg time (s))} & \multirow{2}{*}{$\times$} \\ \cline{7-8}
&      &        &         &      &   & no safety & safety &  \\ \hline

\multirow{3}{*}{\rotatebox{90}{\shortstack{\textbf{Zebrafish}\\\textbf{}}}}& 1-3 & 15405 & 14 (210) & 0.003 & 87.5 & 15405 (0.028) & 15405 (0.016) & 1.5 \\
& 4-6 & 239 & 35 (156) & 0.008 & 29.5 & 225 (15.021) & 239 (0.565) & 80.5 \\
& 7-9 & 6 & 70 (158) & 0.023 & 15.3 & 0 (-) & 6 (5.883) & 92.5 \\
& 10+ & 1 & 81 (81) & 0.030 & 13.1 & 0 (-) & 1 (33.468) & 9.0 \\

\hline

\multirow{3}{*}{\rotatebox{90}{\shortstack{\textbf{Human}\\\textbf{}}}}& 1-3 & 10729 & 15(331) & 0.003 & 78.9 & 10729 (0.046) & 10729 (0.024) & 2.0 \\
& 4-6 & 947 & 35(163) & 0.008 & 24.7 & 861 (16.873) & 947 (0.642) & 63.0 \\
& 7-9 & 92 & 54(105) & 0.012 & 14.4 & 4 (69.780) & 84 (18.858) & 112.2 \\
& 10+ & 12 & 80(104) & 0.017 & 11.1 & 0 (-) & 4 (128.223) & 21.0 \\

\hline

\multirow{3}{*}{\rotatebox{90}{\shortstack{\textbf{Mouse}\\\textbf{}}}}& 1-3 & 12280 & 14 (128) & 0.003 & 84.4 & 12280 (0.029) & 12280 (0.015) & 1.7 \\
& 4-6 & 749 & 35 (126) & 0.007 & 27.4 & 700 (15.287) & 749 (0.691) & 63.7 \\
& 7-9 & 60 & 58 (231) & 0.012 & 16.2 & 0 (-) & 48 (28.615) & 108.8 \\
& 10+ & 9 & 79 (121) & 0.012 & 14.7 & 0 (-) & 4 (37.326) & 19.3 \\

\hline

\multirow{3}{*}{\rotatebox{90}{\shortstack{\textbf{SRR020730}\\\textbf{}}}}& 1-3 & 35069 & 9(186) & 0.002 & 84.7 & 35069 (0.026) & 35069 (0.015) & 1.7 \\
& 4-6 & 4497 & 33(171) & 0.007 & 20.4 & 4429 (13.187) & 4497 (0.449) & 30.0 \\
& 7-9 & 1008 & 52(131) & 0.012 & 11.9 & 265 (125.498) & 998 (9.960) & 182.5 \\
& 10+ & 296 & 75(184) & 0.018 & 8.3 & 1 (7.211) & 202 (42.868) & 88.2 \\

\hline

\multirow{3}{*}{\rotatebox{90}{\shortstack{\textbf{Mouse}\\\textbf{ONT}}}}& 1-3 & 18527 & 11 (129) & 0.002 & 81.4 & 18527 (0.014) & 18527 (0.010) & 1.5 \\
& 4-6 & 3083 & 31 (387) & 0.006 & 30.2 & 3079 (1.401) & 3083 (0.050) & 21.5 \\
& 7-9 & 762 & 48 (157) & 0.010 & 18.1 & 700 (16.120) & 760 (0.228) & 459.0 \\
& 10+ & 442 & 87 (589) & 0.019 & 11.1 & 228 (35.735) & 433 (1.029) & 1182.1 \\

\hline

\multirow{3}{*}{\rotatebox{90}{\shortstack{\textbf{Mouse}\\\textbf{ONT.tr}}}}& 1-3 & 18029 & 12(131) & 0.003 & 81.1 & 18029 (0.019) & 18029 (0.014) & 1.6 \\
& 4-6 & 3028 & 32(187) & 0.007 & 30.6 & 3012 (1.681) & 3028 (0.065) & 21.1 \\
& 7-9 & 803 & 50(419) & 0.012 & 19.0 & 728 (15.169) & 800 (0.371) & 238.9 \\
& 10+ & 476 & 89(549) & 0.022 & 11.6 & 250 (31.191) & 458 (1.346) & 683.3 \\

\hline

\multirow{3}{*}{\rotatebox{90}{\shortstack{\textbf{Mouse}\\\textbf{PacBio}}}}& 1-3 & 14256 & 14(382) & 0.004 & 83.5 & 14256 (0.035) & 14256 (0.018) & 2.2 \\
& 4-6 & 1376 & 33(187) & 0.008 & 33.3 & 1356 (2.734) & 1376 (0.116) & 41.7 \\
& 7-9 & 182 & 49(159) & 0.012 & 20.5 & 144 (35.090) & 182 (1.531) & 489.7 \\
& 10+ & 63 & 109(1178) & 0.017 & 14.2 & 14 (100.395) & 56 (0.640) & 1240.8 \\

\hline

\end{tabular}
\end{center}
\end{table}

\end{document}